\documentclass[12pt]{article}

\usepackage[small]{titlesec}
\usepackage{
    amsmath,
    amssymb,
    amsthm,
    graphicx,
    setspace,
    booktabs,
    subcaption,
    float
}
\usepackage{breakcites}
\usepackage{enumitem}

\usepackage[
    colorlinks=true,
    allcolors=LinkBlue,
    backref=page
]{hyperref}

\usepackage{lmodern}
\usepackage{multirow}
\usepackage[square,sort,comma,numbers]{natbib}
\usepackage{titling}
\usepackage{xcolor}
\usepackage{algpseudocode}

\allowdisplaybreaks

\AtBeginDocument{
    \abovedisplayskip=5pt plus 2pt minus 2pt
    \belowdisplayskip=\abovedisplayskip
    \abovedisplayshortskip=2pt plus 2pt minus 2pt
    \belowdisplayshortskip=\belowdisplayskip
}

\titlespacing*{\section}{0pt}{*2}{*1}
\titlespacing*{\subsection}{0pt}{*2}{*1}

\bibpunct{(}{)}{;}{a}{}{,}
\setlist{noitemsep,topsep=0pt}
\definecolor{LinkBlue}{rgb}{.15,.25,.85}

\binoppenalty=\maxdimen
\relpenalty=\maxdimen

\makeatletter
\renewcommand*{\NAT@spacechar}{~}
\makeatother

\usepackage{algorithm}
\makeatletter

\newfloat{algorithm}{tbp}{loa}
\providecommand{\algorithmname}{Algorithm}
\floatname{algorithm}{\protect\algorithmname}

\makeatother

\newtheorem{theorem}{Theorem}

\newtheorem{lemma}{Lemma}

\newtheorem{prop}{Proposition}

\def \bP {\mathbb{P}}
\def \bE {\mathbb{E}}

\def \bN {\mathbb{N}}

\def \row {\mathbf{r}}
\def \col {\mathbf{c}}
\def \rmax {\row_\mathsf{max}}
\def \rmin {\row_\mathsf{min}}
\def \cmax {\col_\mathsf{max}}
\def \cmin {\col_\mathsf{min}}

\usepackage{xspace}

\newcommand{\TV}{{\sf TV}}
\newcommand{\gap}{\operatorname{gap}}

\newcommand{\Bern}{\text{Bern}}

\newcommand{\Binom}{{\mathsf{Binom}}}

\newcommand{\indc}[1]{{\mathbf{1}_{\left\{{#1}\right\}}}}
\newcommand{\Indc}{\mathbf{1}}

\newcommand{\calO}{{\mathcal{O}}}

\newcommand{\mix} {\textsc{mix~}}

\newcommand{\snakep} {\textsc{Snake}^{+}}
\newcommand{\swap} {\textsc{Swap~}}

\newcommand{\snake} {\textsc{Snake~}}

\newcommand{\dsnake} {D-\textsc{Snake~}}

\newcommand{\tin} {\text{in}}
\newcommand{\out} {\text{out}}
\newcommand{\biparG} {G_{\text{Bi}}}
\newcommand{\bisimpG} {G_{\text{Bi,sim}}}

\begin{document}

\def\spacingset#1{\renewcommand{\baselinestretch}{#1}\small\normalsize} \spacingset{1}

{
  \title{\bf The Snake Algorithm: A Rejection-Free Sampler for Binary Matrices with Fixed Margins}
 \author{
    Zipei Nie\thanks{Department of Mathematics, University of Illinois Urbana-Champaign,
    Urbana, IL 61801, USA.
    Email: \texttt{znie@illinois.edu}}
    \and
    Guanyang Wang\thanks{Department of Statistics, Rutgers University--New Brunswick,
    Piscataway, NJ 08854, USA.
    Email: \texttt{guanyang.wang@rutgers.edu}}
    \and
    Peng Zhang\thanks{Department of Computer Science, Rutgers University--New Brunswick,
    Piscataway, NJ 08854, USA.
    Email: \texttt{pz149@rutgers.edu}}
  }

  \maketitle
}

\bigskip
\begin{abstract}
We study uniform sampling of binary matrices with fixed row and column sums, a recurring problem in ecological null models, Rasch-model testing, network analysis, and combinatorics. We propose the \snake algorithm, a rejection-free Markov chain Monte Carlo sampler that grows an alternating path until its first self-intersection and flips the resulting loop. The chain is reversible and irreducible on the fixed-margin state space, hence has the uniform stationary distribution. We prove that one step flips on the order of $\sqrt n$ entries in sparse and balanced square regimes, give upper bounds on the per-step path length, and show that the resulting work per flipped entry is rate optimal in sparse and balanced regimes and near-optimal up to a polylogarithmic factor under a one-sided half-balanced condition. A Markov-chain comparison, combined with the recently established universal spectral-gap bound for the swap chain, proves that the lazy \snake chain is rapidly mixing for every feasible pair of margins; in the permutation-matrix case, the raw chain has the sharp total-variation mixing time $\Theta(\sqrt n\log n)$. We also describe a directed-graph extension and an equal-margin label-shuffling variant.  Numerical experiments against Swap, Rectangle Loop, Curveball, sequential importance sampling, and a directed edge-swap algorithm show consistent gains in move size, wall-clock convergence, and sampling efficiency.
\end{abstract}

\vfill

\newpage
\spacingset{1.9} \section{Introduction}\label{sec:intro}
Binary matrices are commonly used across multiple scientific domains, such as ecology, psychology, statistics, operations research, and computer science, to represent binary relations between two sets of objects. A key problem in this field is generating binary matrices with fixed row and column sums, which is at the heart of many hypothesis testing problems. 

For example, researchers use a binary matrix to represent the species--habitat combinations when studying species co-occurrence. A one (resp. zero) in the $(i,j)$-th entry corresponds to the presence (resp. absence) of the $i$-th species in the $j$-th habitat (Figure \ref{fig:bipartite-matrix}). Given an observed co-occurrence matrix, ecologists have long debated whether certain observed patterns reveal that the competition between species influences their distribution or they can simply be explained by random chance. To answer the question, they compare metrics in the observed matrix with those in uniformly sampled binary matrices with observed row and column sums (a null model). Fixing the row and column sums accounts for natural variation in species distribution and resource availability in habitats \citep{joppa2010nestedness}. 
\begin{figure}[H]
		\centering
		\includegraphics[scale=0.32]{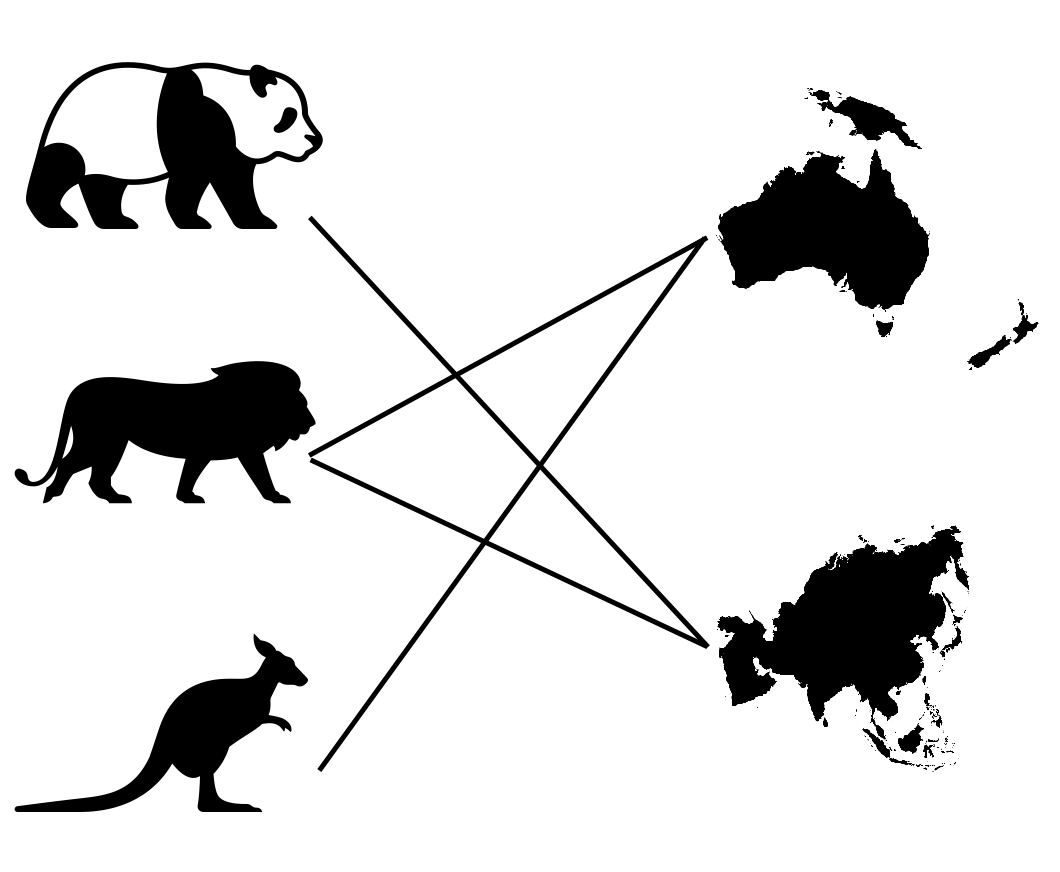}\qquad
		\includegraphics[scale = 0.3]{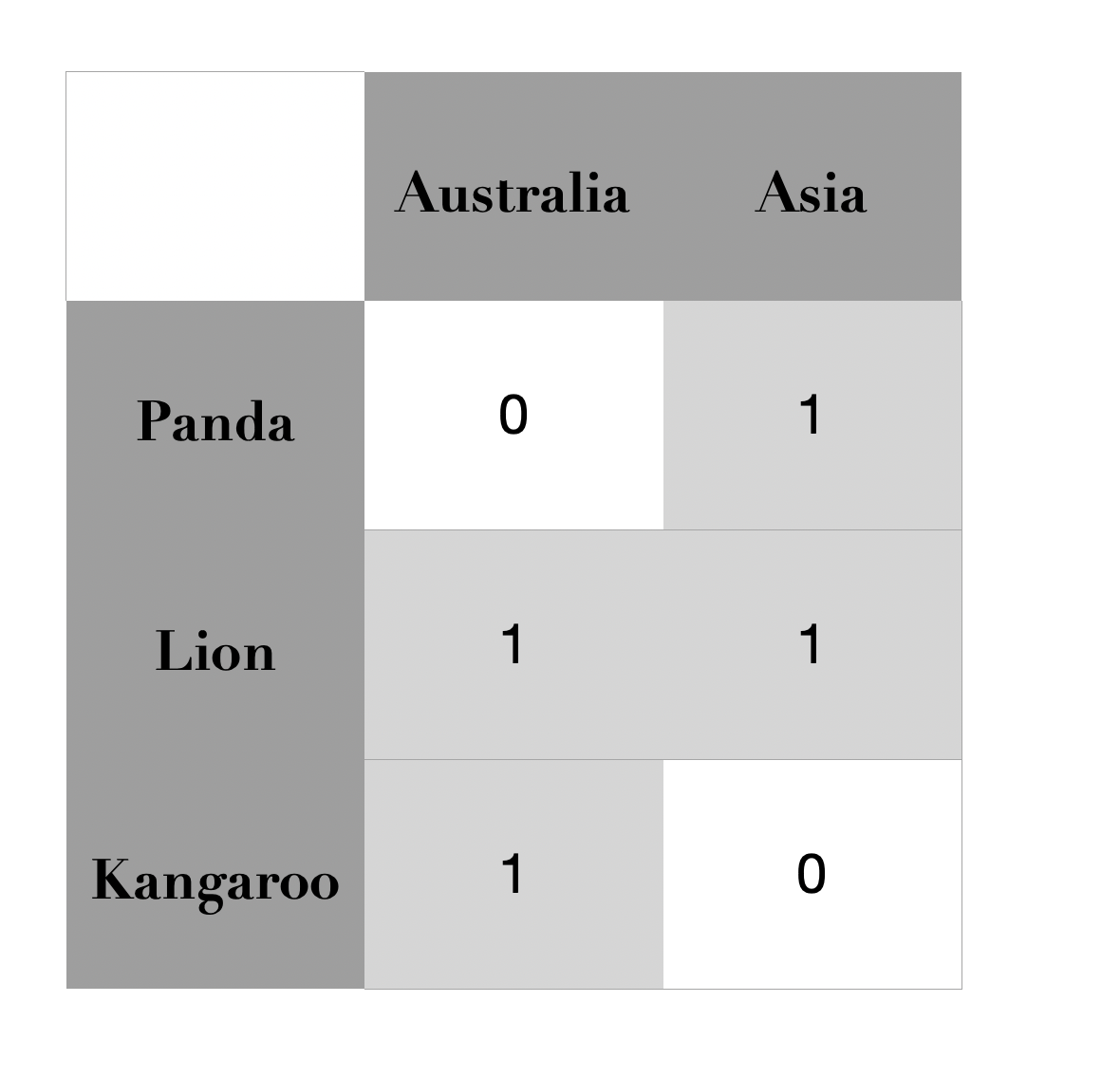}
		\caption{An example of the species-habitat relation and its binary matrix representation. Left: bipartite representation with $m = 3$ species and $n = 2$ habitats. An edge represents the presence of one species in the corresponding habitat. Right: binary matrix representation. Each row corresponds to a species, and each column corresponds to a habitat.}
		\label{fig:bipartite-matrix}
	\end{figure}

Another important application is in testing the Rasch model \citep{rasch1993probabilistic}. This model assumes that the entries in a binary matrix $M$ are mutually independent, and the probability of a cell $(i,j)$ being $1$ is given by 
\begin{align*}
    \bP(M_{i,j} = 1) = \frac{\exp(\theta_i - \beta_j)}{1 + \exp(\theta_i - \beta_j)},
\end{align*}
where $\theta_i$ and $\beta_j$ are model parameters. The Rasch model is widely used in various fields, such as psychometrics, education, agriculture, and healthcare. For instance, in a test with $m$ students and $n$ problems, the $(i,j)$-th entry of $M$ is 1 if the $i$-th student answers the $j$-th problem correctly and 0 otherwise. Here, $\theta_i$ represents the ability of the $i$-th student, and $\beta_j$ denotes the difficulty of the $j$-th problem. Notably, all binary matrices with fixed row and column sums have the same probability under the Rasch model. Therefore, efficient uniform sampling algorithms can be used for reliable goodness-of-fit tests of the model \citep{besag1989generalized, glas1995testing, ponocny2001nonparametric}. 

Sampling binary matrices with fixed row and column sums also has applications in network reconstruction from partial information (e.g., interbank networks) \citep{glasserman2022maximum}, differential privacy (e.g., releasing permuted data while maintaining aggregated statistics) \citep{gong2020congenial}, and theoretical computer science (e.g., transforming undirected sparse graphs to expanders) \citep{allen2016expanders,giakkoupis2022expanders}.

One possible way to sample a binary matrix with fixed row and column sums is to enumerate all matrices that satisfy the margin constraints and sample uniformly from this set. However, the number of binary matrices with fixed row and column sums usually grows exponentially with the matrix's dimensions, making it impractical to calculate the sample space's size \citep{barvinok2010number,barvinok2012matrices}. Therefore, it is typically not feasible to exhaustively enumerate all the matrices that satisfy the margin constraints.

Practically-efficient algorithms  can be broadly classified into two categories. The first group of algorithms generates independent samples which are not necessarily uniformly distributed. Examples include \cite{chen2005sequential, miller2013exact} and the references therein. These non-uniformly distributed samples can be used to estimate many statistics of interest such as the cardinality of the matrices with given margin sums. The second category of algorithms includes Markov Chain Monte Carlo (MCMC) methods with the uniform fixed-margin distribution as their target; after aperiodicity corrections such as laziness, their draws converge asymptotically to that target. In this category, the swap algorithm \citep{besag1989generalized, diaconis1998algebraic} is studied most extensively. Other methods \cite{rao1996markov, verhelst2008efficient,strona2014fast,wang2020fast} have been developed with limited theoretical results. Many classical swap- and rectangle-loop-type methods make local flips on the space of binary matrices with fixed margins, which can become inefficient when the row and column sums are sparse.  For example, the swap algorithm takes $\Omega(mn)$ steps to flip four entries of an $m\times n$ sparse matrix. On the other hand, most large binary matrices appearing in practice are sparse \citep{vasques2020transitivity}. This motivates us to design  efficient algorithms that can handle high-dimensional and sparse matrices.

Unlike local-flip methods, our new algorithm finds a random, swappable loop at every step and flips every entry in the loop. The \snake algorithm has two new features: it is rejection-free, which guarantees successful flipping at every iteration, and the number of flipped entries typically increases with the matrix's size. In the case of sparse matrices, one step of the algorithm can flip on the order of $\sqrt{\min\{m,n\}}$ entries, substantially larger than the $O(1)$ or $O(1/mn)$ entries flipped by swap- and rectangle-loop-type methods. For general fixed margins, our comparison with the classical swap chain, together with the universal swap-chain spectral-gap bound of \citet{fu2026spectral}, proves that the lazy \snake algorithm is rapidly mixing for every feasible pair of margins. We also prove a sharper, problem-specific result in the extreme sparse case $\row=\col=\mathbf{1}_n$, where the state space is the symmetric group $S_n$ and the raw \snake chain mixes in $\Theta(\sqrt n\log n)$ steps. Additionally, we discuss two natural variants:
\begin{itemize}
    \item We evaluate an equal-margin label-shuffling variant, denoted $\snakep$. It interleaves ordinary \snake moves with random permutations of rows and columns within equal-margin classes. A related symmetry strategy has been used in \citep{miller2013exact} for counting matrices with fixed margins.
    \item Since every binary matrix is the adjacency matrix of a bipartite graph (see Figure \ref{fig:bipartite-matrix} for an example), sampling binary matrices with fixed row and column sums is equivalent to sampling bipartite graphs with fixed degree sequences. Similarly, every square matrix with zero diagonal entries is the adjacency matrix of a directed graph. Hence, by making a simple modification to the original algorithm, we propose the Directed-Snake (\dsnake) algorithm for sampling simple directed graphs with fixed in- and out-degree sequences.
\end{itemize}

After a brief discussion of the existing methods and their limitations, we introduce our algorithm and its variants in Section \ref{sec: Snake algorithm}. Section \ref{sec: Snake algorithm} gives the construction and the main implementation details. Section \ref{sec:theoretical-analysis} gives theoretical guarantees for correctness, the number of flipped entries per step, computational cost, and mixing, including the sharp permutation-matrix bound. The full proofs are given in Appendix \ref{app:proofs}. Section \ref{sec:numerical} presents numerical experiments comparing our method with existing MCMC and sequential importance sampling methods.

\section{The Snake Algorithm}\label{sec: Snake algorithm}

To start, we introduce the notations that will be used throughout the paper. We preserve the letters $m$ and $n$ to denote the number of rows and columns of the matrix, respectively. We denote the $i$-th row of the matrix as $R_i$ (upper-case), and its corresponding row sum as $r_i$ (lower-case). Similarly, we use $C_j$ (upper-case) to label the $j$-th column and $c_j$ (lower-case) to denote its sum. Given row sums $\row = (r_1, r_2, \ldots, r_m) \in \bN^m$ and column sums $\col = (c_1, c_2, \ldots, c_n) \in \bN^n$, we use $\Sigma(\row, \col)$ to denote the set of all binary matrices with row sums $\row$ and column sums $\col$. The celebrated Gale-Ryser Theorem \citep{gale1957theorem,ryser1957combinatorial} provides necessary and sufficient conditions for $\Sigma(\row, \col)$ to be non-empty and gives an algorithmic way to generate an instance from $\Sigma(\row, \col)$. We denote the uniform distribution on $\Sigma(\row, \col)$ by $U(\Sigma(\row, \col))$. Without loss of generality, we assume that each $r_i$ is between $1$ and $n-1$ and each $c_j$ is between $1$ and $m-1$ because, in the case where they are not, the corresponding row/column is completely determined by its summation, and the user can delete the degenerate row/column. We also define $\rmax := \max_{1\le i\le m} r_i,$ $\rmin:= \min_{1\le i\le m} r_i,$ $\cmax := \max_{1\le j\le n} c_j,$ and $\cmin:= \min_{1\le j\le n} c_j$ for the maximum row sum, the minimum row sum, the maximum column sum, and the minimum column sum, respectively.
	
Given two positive sequences $a(n),b(n)$, we say $a(n) \asymp b(n)$ or $a(n) = \Theta(b(n))$ if there are positive constants $c_1, c_2 > 0$ such that $c_1 b(n) < a(n) < c_2 b(n)$. We say $a(n) = \Omega(b(n))$ if $a(n) > c_1 b(n)$, and $a(n) = \calO(b(n))$ if $a(n) < c_2 b(n)$. Constants in asymptotic notation may depend on fixed regularity parameters such as degree bounds, balance constants, comparability constants in $m\asymp n$, and fixed error tolerance, but not on the growing matrix dimensions unless stated otherwise. For each positive integer $k$, we write $[k]$ for the set $\{1,\ldots, k\}$. Fix a matrix $M$, we write $M(i,j)$ for the $(i,j)$-entry of $M$, $M(i,\cdot)$ for the $i$-th row and $M(\cdot, j)$ for the $j$-th column. Given an ordered/unordered set $A$ in $[m]$, we write $M[A, \cdot]$ for the ordered/unordered list of rows in $M$ with indexes in $A$. Similarly, we define $M[\cdot, B]$ in the same way for $B$ in $[n]$.
\subsection{Motivation: From swap to loop-finding}
Starting from an initial matrix, $M_0 \in \Sigma(\row,\col)$, each step of the swap algorithm uniformly samples two rows and two columns without replacement. If the resulting two by two matrices is a checkerboard unit, i.e.,
\begin{align*}
     \left(
\begin{array}{cc}
    	1 & 0 \\
		0 & 1 
\end{array} \right)  \qquad \text{or} \qquad
\left(
\begin{array}{cc}
    	0 & 1 \\
		1 & 0
\end{array} \right),
\end{align*}
then the algorithm swaps one checkerboard to the other; otherwise, do nothing. This transformation  always maintains the margin sums. It is also proven that this algorithm defines a reversible Markov chain on $\Sigma(\row,\col)$ with uniform stationary distribution \citep{diaconis1998algebraic}.

The swap algorithm, while mathematically correct, may not be efficient for practical use. Two observations support this notion. First, there is a small number of flipped entries: During each step, the matrix $M_t$ either remains the same as $M_{t-1}$ or differs by four entries in a $2\times 2$ matrix. This means that regardless of the matrix size, the maximum number of flipped entries is four. Second, there is a high rejection rate: Swapping only occurs if the randomly selected four entries consist of exactly two ones and two zeros. When dealing with sparse or full original matrices, it may take around $m\times n$ iterations to achieve a single successful swap. These two facts indicate that the algorithm's scalability with matrix size is poor. Considering the permutation matrices as an extreme example. The swap algorithm takes $(n^2-n)/2$ iterations in expectation to make one single swap. Meanwhile, the swap algorithm on permutation matrices is equivalent to a $(1 - 1/\binom{n}{2})$-lazy random transposition walk on the permutation group $S_n$. It follows from Markov chain literature \citep{diaconis1981generating} that this algorithm has a mixing time of $n^3\log n/4 + cn^3$. Therefore, it takes around $1.15\times 10^6$ steps to mix a $100\times 100$  matrix; and more than $1.72 \times 10^9$ steps to mix a $1000\times 1000$ matrix.  These examples highlight the scalability issues  with increasing matrix size.

Based on this observation, the rectangle loop algorithm \citep{wang2020fast} aims to improve the mixing time, as illustrated in Figure \ref{fig:swap-rectangle}. Instead of  simultaneously selecting two rows and two columns, each step of the rectangle loop algorithm first chooses one row and one column, assuming the coordinate is $(R_2, C_2)$. Assuming the corresponding random element equals $1$,  the algorithm then samples a $0$ (assuming in column $C_4$) among all the $0$s in the same row. Then the algorithm samples a $1$ (assuming in row $R_5$) among all the $1$s in $C_4$. Finally, one checks the entry $(R_5,C_2)$. If it is $0$, then a checkerboard is found, and we swap it. Otherwise, we do nothing. Details of the algorithm can be found in \cite{wang2020fast}.

\begin{figure}[H]
		\centering
		\includegraphics[width=\textwidth]{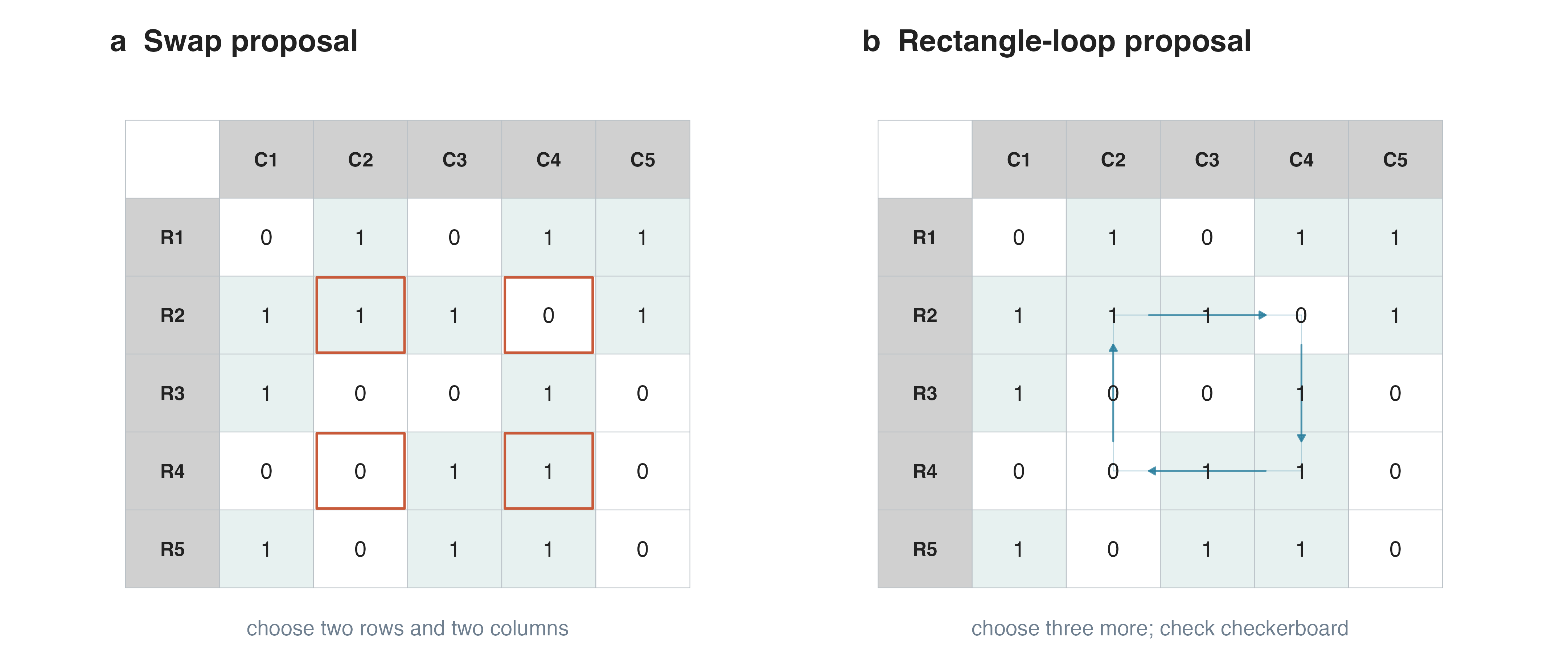}
		\caption{Comparison between swap and rectangle-loop proposals. Panel a shows the classical swap proposal, which chooses two rows and two columns simultaneously and succeeds only if the induced $2\times 2$ block is a checkerboard. Panel b shows the rectangle-loop proposal, which conditions each subsequent choice on the previous row or column and then checks whether the final cell closes a checkerboard.}
		\label{fig:swap-rectangle}
	\end{figure}

Despite the higher swap probability, the rectangle loop algorithm still often fails when applied to sparse or full matrices. Additionally, even upon successful execution, it still only flips four entries. To overcome these limitations, our new algorithm described below generates longer loops that guarantee the successful flipping of all entries within the loop.

\subsection{The Snake Algorithm}
Given a binary matrix $M$, we define an alternating loop (or loop) as an ordered sequence of distinct entries $\{(R_{i_1}, C_{i_1}), (R_{i_2}, C_{i_2}), \ldots, (R_{i_{2k}}, C_{i_{2k}})\}$ satisfying the following two conditions:
\begin{itemize}
    \item Consecutive entries have alternating values, i.e.,:
    $$M(R_{i_1}, C_{i_1}) = 1 - M(R_{i_2}, C_{i_2}) = M(R_{i_3}, C_{i_3}) =  1 - M(R_{i_4}, C_{i_4}) \ldots   $$
    \item Consecutive entries are in the same row or column alternatively, i.e., the entries satisfy one of the two patterns:
   \begin{itemize}
       \item  $R_{i_1} = R_{i_2}, ~C_{i_2} = C_{i_3},~ R_{i_3} = R_{i_4}, ~ C_{i_4} = C_{i_5},  \ldots, ~C_{i_{2k}} = C_{i_1}$, or
       \item $C_{i_1} = C_{i_2}, ~R_{i_2} = R_{i_3},~ C_{i_3} = C_{i_4}, ~ R_{i_4} = R_{i_5}, \ldots, ~R_{i_{2k}} = R_{i_1} $
   \end{itemize}
\end{itemize}

An oriented checkerboard is an example of a loop with a length of $4$. However, loops can have a much longer length. For instance, in Figure \ref{fig:different loops}, there are three different loops with lengths $4, 6,$ and $8$.

\begin{figure}[!htbp]
		\centering
		\includegraphics[width=0.95\textwidth]{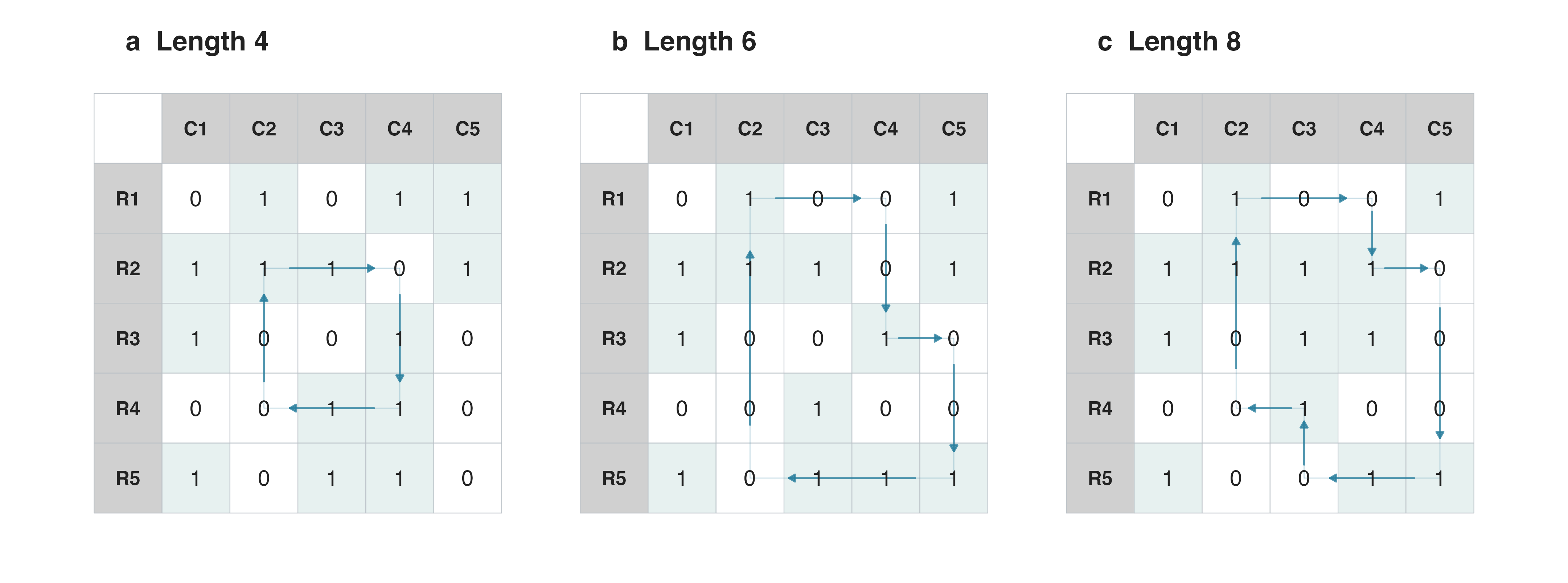}
		\caption{Alternating loops of lengths $4,6$, and $8$. Each highlighted cycle alternates between one- and zero-entries and between row and column moves, so flipping the highlighted cells preserves every row and column sum.}
		\label{fig:different loops}
	\end{figure}

A key feature of the alternating loop is that flipping all the entries in the loop will not change the row and column sums. Therefore, the main idea of our \snake algorithm is to find an alternating loop, and then flip all the entries. Our algorithm is  described below:
\begin{algorithm}[htbp]
	\caption{One step of the \snake algorithm}\label{alg:Snake}
	
	\begin{algorithmic}
		
		\State \textbf{Input:} 
		\begin{itemize}
			\item Row sums $\row$ and column sums $\col$
			\item Current state $M \in \Sigma(\row,\col)$
			\item Trajectory $T = \varnothing $
		\end{itemize}
		\State Choose one row and one column $S_1 = (R_1, C_1)$ uniformly at random, add $S_1$ into $T$
		\State Set $K = 1$
        \While{\textsc{True}}
         \If{$M(R_K, C_K) = 1$}
         \State Sample $C_{K+1}$ uniformly from the set $\{1\leq j \leq n: M(R_K, j) = 0\}$
         \State Add $S_{K+1} = (R_K, C_{K+1})$ into $T$
         \Else
         \State Sample $R_{K+1}$ uniformly from the set $\{1\leq i \leq m: M(i, C_K) = 1\}$
          \State Add $S_{K+1} = (R_{K+1}, C_{K})$ into $T$
         \EndIf
         \State $K \gets K+1$
         \If{$(S_{K'}, S_{K'+1}, \ldots, S_{K})$ forms an alternating loop for some $1\leq K'< K$}
         \State Update $(M(S_{K'}), M(S_{K'+1}), \ldots, M(S_{K})) \gets 1 - (M(S_{K'}), M(S_{K'+1}), \ldots, M(S_{K})) $
         \State \textbf{Break}
         \EndIf

		\EndWhile

		\State\textbf{Return: The updated $M$} 
	\end{algorithmic}
\end{algorithm}

The algorithm is called \snake because each step is similar to the classical video game Snake. In the game, the snake continually gets longer as it moves. The game finishes when the snake runs into itself, i.e., forms a loop. Figure \ref{fig:snake} provides an illustrative example using $5\times 5$ binary matrices.

\begin{figure}[H]
		\centering
		\includegraphics[width=0.95\textwidth]{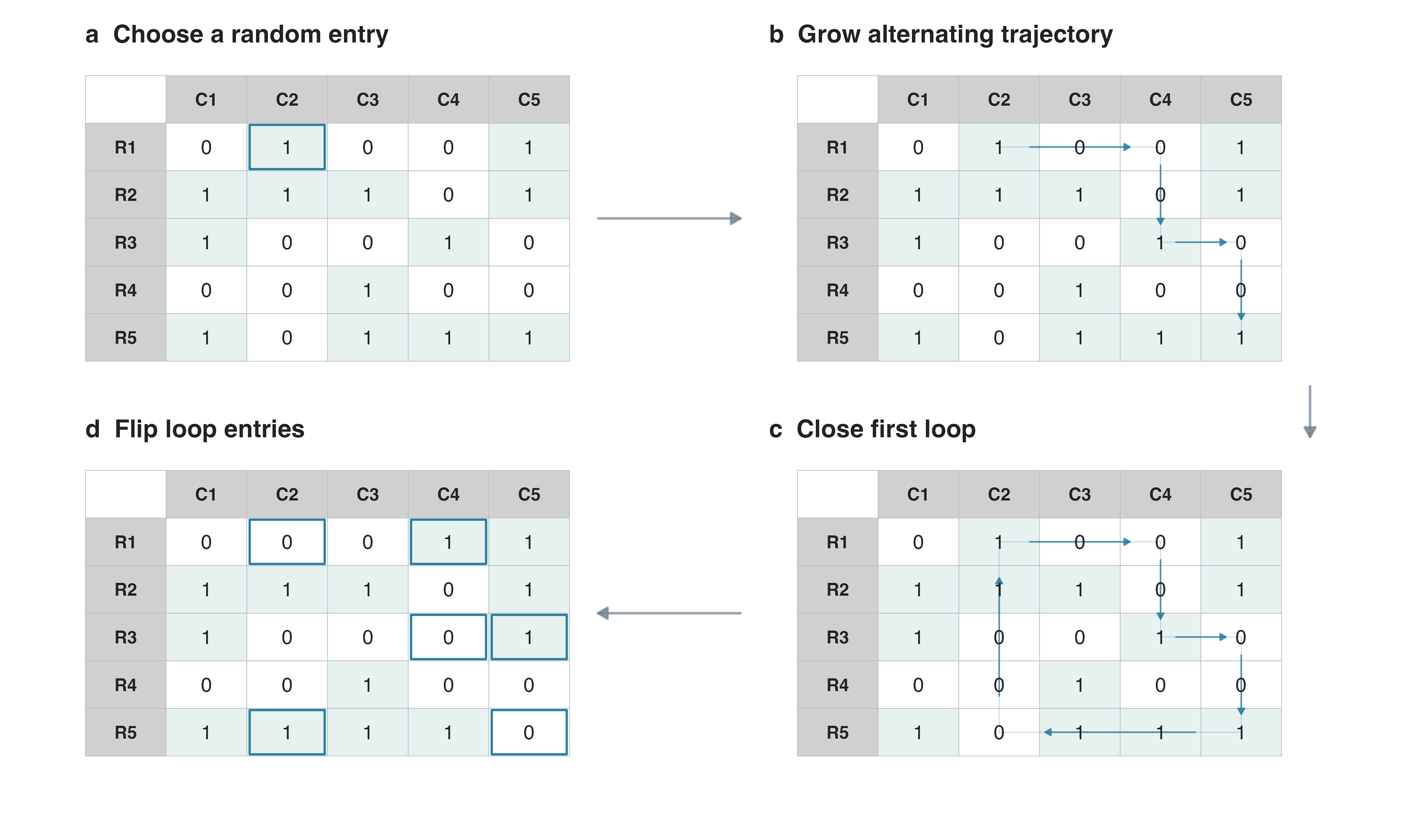}
		\caption{One transition of Algorithm \ref{alg:Snake}. Starting from a uniformly chosen entry, the trajectory alternates between row and column moves through entries of opposite value until the first self-intersection closes an alternating loop. The algorithm then flips all entries on that loop.}
		\label{fig:snake}
	\end{figure}

Now we dive into Algorithm \ref{alg:Snake} and explain why it is rejection-free (or why it will always find a loop). Given a trajectory $(S_1, S_2, \ldots, S_T)$ without loop, let us assume without loss of generality that $M(S_T) = 1$. Then $S_{T+1} := (R_T, C_{T+1})$ shares the same $x$-coordinate with $S_T$. Meanwhile, we can observe that the new trajectory $(S_1, S_2, \ldots, S_{T+1})$ contains a loop if and only if there is some $T'\in \{1,2,\ldots, T\}$ such that $S_{T'}$ has the same $y$-coordinate as $S_{T+1}$. However, there are only finitely many rows and columns. Therefore the while loop will terminate after a finite number of steps; equivalently, a repeat occurs after at most $m+n+1$ visited entries. Meanwhile, the alternating loop will always be unique when the algorithm stops. 

Next, we briefly discuss some theoretical properties and implementation details of the \snake algorithm. Most discussions will be informal; formal statements and proofs are given in Section \ref{sec:theoretical-analysis} and Appendix \ref{app:proofs}. We start with the correctness of the \snake algorithm. Let $\bP_\snake$ be the Markov transition kernel described by Algorithm \ref{alg:Snake}. One can show $\bP_\snake(\cdot,\cdot)$ is a symmetric function on $\Sigma(\row,\col)\times \Sigma(\row,\col)$. This fact, in turn, implies the \snake algorithm is a reversible, irreducible Markov chain that leaves $U(\Sigma(\row,\col))$ as stationary distribution. 

Then we  discuss the number of flipped entries per step. Denote by $C$ the number of inner-loop iterations for one step of the \snake algorithm and by $L$ the number of flipped entries. The value of $L$ is a direct move-size diagnostic: larger values mean that a single transition makes a more nonlocal change, although actual mixing remains margin- and statistic-dependent.
Though $L=O(C)$ and $C\leq m+n+1$, getting good estimates on the order of $L$ and $C$ is  challenging due to the complex combinatorial structure of $\Sigma(\row,\col)$. The subsequent estimates are derived from a series of combinatorial and probability analyses, which are much sharper than the na\"ive bounds:
\begin{itemize}
\item (No assumption) For any  $\row, \col$ such that $\Sigma(\row,\col) \neq \varnothing$, $\bE(L) = \calO(\sqrt{m+n}\log(m+n))$.
    \item (Sparse) When the row sums are relatively sparse in the sense that $\rmax = \calO(1)$, we have $\bE(L) = \Omega(\sqrt n)$. Similarly  $m-\cmin = \calO(1)$ implies $\bE(L) = \Omega(\sqrt m)$.
    \item (Relatively balanced) When every active column has at least a fixed fraction of ones and every active row has at least a fixed fraction of zeros, in the sense that $m \asymp n, \cmin \geq \delta m, \rmax \leq \Delta n$ for $0 <\delta <\Delta < 1$, we have  $\bE(L) \asymp \bE(C) \asymp \sqrt n$.
    \item (Half-balanced) When either every active row has fewer than half ones, $\rmax\le \delta n$ with $\delta<1/2$, or every active column has more than half ones, $\cmin\ge \Delta m$ with $\Delta>1/2$, the work per flipped entry is near-optimal up to a polylogarithmic factor.
\end{itemize}

Therefore,  the number of flipped entries grows at the rate of $\sqrt{\min\{m,n\}}$ in the sparse and relatively balanced square regimes covered by the theorems below. In contrast, existing local-flip algorithms flip at most $\calO(1)$ entries, and often only $O(1/(mn))$ or $O(1/\max\{m,n\})$ expected entries per attempted step when the matrix is sparse. This gives the \snake chain substantially larger nonlocal moves in regimes where local methods move rarely. 

Finally, we comment on the computational cost of Algorithm \ref{alg:Snake}. Suppose Algorithm \ref{alg:Snake} terminates after $C$ inner-loop iterations. With dense indexed rows and columns, or maintained candidate lists for the row-zero and column-one sampling sets, the sampling plus swapping cost is $\Theta(C)$. It remains to analyze the cost of checking whether a loop exists at each iteration. One can maintain two auxiliary hash tables recording the previously seen $x$ and $y$ coordinates and their appearance times. Since checking the existence of a loop is equivalent to checking if a coordinate has appeared before, at each iteration the implementation needs at most one search and one insert in the current hash table, both $\calO(1)$ operations. Under this implementation model, the total computational cost for running one step of the \snake algorithm once is $\calO(C)$.

\subsection{Extensions}\label{subsec: extension}

\subsubsection{Directed graphs with fixed degree sequence}
Our method can be naturally generalized to sample simple directed graphs with fixed degree sequences, a problem that finds widespread applications in network analysis \citep{kim2012constructing,berger2010uniform}. Given a simple directed graph $G$ with $n$ nodes labelled by $\{1,2,\ldots, n\}$, its adjacency matrix is a  square matrix $A$ with $A(i,j) = 1$ if there is an edge from $i$ to $j$, and $0$ otherwise.  The diagonal elements in $A$ are all zero as there is no self-loop. The sums of the rows, denoted as $r$, and the sums of the columns, denoted as $c$, of matrix $A$ correspond to the out-degree and in-degree sequences of graph $G$, respectively. Therefore, sampling simple directed graphs with fixed in and out-degree sequences is equivalent to sampling binary matrices with fixed margin sums and zero diagonal entries. Towards this goal, we propose the \dsnake algorithm based on the \snake algorithm, which we now describe in detail. 

\begin{figure}[H]
		\centering
		\includegraphics[width=\textwidth]{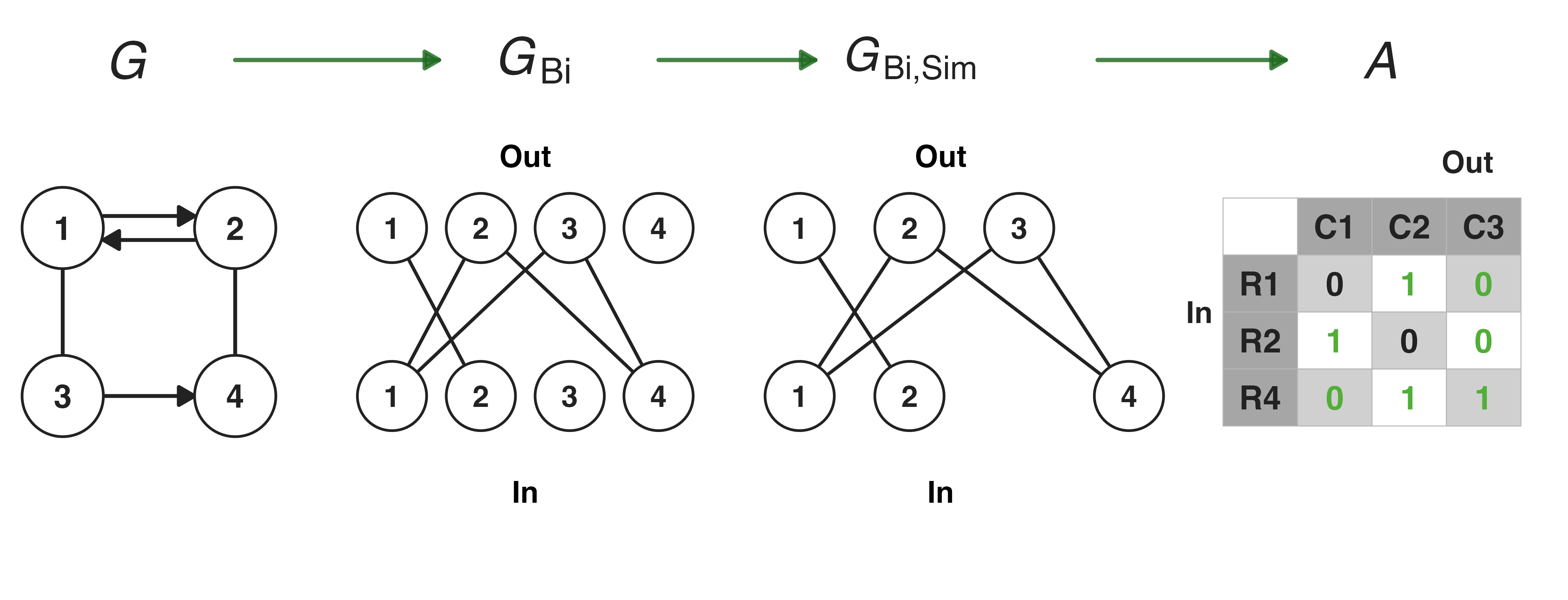}
		\caption{Preprocessing for the directed-graph extension. A simple directed graph is represented as a bipartite graph between out-labels and in-labels, redundant labels are removed, and the remaining fixed-degree problem is represented as a binary matrix with label-diagonal cells forbidden. The black diagonal zero entries mark structural no-self-loop constraints that cannot be flipped during sampling.}
		\label{fig:preprocessing}
	\end{figure}

Given a directed graph $G$, we first represent it using a bipartite graph $\biparG$ with $2n$ nodes, $n$ nodes in each part. Every node $i$ in $G$ is associated with two nodes $i_\tin$ and $i_\out$ in the two parts of $\biparG$. Any directed edge $i\rightarrow j$ in $G$ is then represented by an undirected edge connecting $i_\out$ and $j_\tin$. Now every node in $\biparG$ has at most $n-1$ edges, and there is no edge connecting $i_\out$ and $i_\tin$ for every $i\in G$. Meanwhile, there may be redundant nodes in the sense that they connect to zero or all the other $(n-1)$ nodes in the other part. The removal is performed iteratively. Whenever a node is deleted, the residual degrees of the remaining nodes are updated, and the process continues until no redundant nodes remain. After removing the redundant nodes, we have row and column label sets $U,V\subset [n]$ representing the remaining labels in the two parts. The simplified bipartite graph $\bisimpG$ is then a subgraph of $\biparG$, with two parts $U$ and $V$. Finally, we construct the $\lvert U\rvert \times \lvert V \rvert$ dimensional adjacency matrix of $\bisimpG$. In particular, we  index every row and column according to the elements in $U$ and $V$ (instead of $[\lvert U\rvert]$ and $[\lvert V \rvert]$). The resulting binary matrix contains all the information for the degree sequences. Every row $i \in U$ has at least one `$1$', and at most $\left(\lvert V \rvert - 1 - \Indc(i \in V)\right)$ `$1$'s (as otherwise it is redundant). Similarly, every column $j \in V$ has at least one `$1$', and at most $\left(\lvert U \rvert - 1 - \Indc(j \in U)\right)$ `$1$'s. The whole procedure is explained via Figure \ref{fig:preprocessing}.

After preprocessing, it suffices to sample binary matrices uniformly with additional constraints. Besides keeping the row and column sums invariant, we need to further ensure the diagonal entries $D := \{(s,s), s\in U \cap V\}$ (e.g., the black entries in Figure \ref{fig:preprocessing}) are all zero throughout the sampling procedure, as we require no self-loops. We write the constrained state space as
\[
    \Sigma_D(\row,\col;U,V)
    =\{M\in\{0,1\}^{U\times V}:\ M\mathbf 1=\row,\ M^\top\mathbf 1=\col,\ M(s,s)=0\ \text{for all }s\in U\cap V\}.
\]
Sampling simple directed graphs with the prescribed in- and out-degrees is equivalent to sampling uniformly from this constrained space after the redundant labels have been removed. This can be achieved by slightly 
modifying the \snake algorithm as follows:
\begin{algorithm}[htbp]
	\caption{One step of the \dsnake algorithm}\label{alg:D-Snake}
	
	\begin{algorithmic}
		
			\State \textbf{Input:} 
			\begin{itemize}
				\item Row labels $U$, row sums $\row$;  column labels $V$, column sums $\col$
				\item Current state $M \in \Sigma_D(\row,\col;U,V)$
				\item Trajectory $T = \varnothing $
			\end{itemize}
			\State Choose one row and one column $X_1 = (R_1, C_1)$ uniformly over $(U\times V) \setminus D$, add $X_1$ into $T$
			\State Set $K = 1$
	        \While{\textsc{True}}
	         \If{$M(R_K, C_K) = 1$}
	         \State Sample $C_{K+1}$ uniformly from the set $\{j\in V: M(R_K, j) = 0,\ j\ne R_K\}$
	         \State Add $X_{K+1} = (R_K, C_{K+1})$ into $T$
	         \Else
	         \State Sample $R_{K+1}$ uniformly from the set $\{i\in U : M(i, C_K) = 1\}$
	          \State Add $X_{K+1} = (R_{K+1}, C_{K})$ into $T$
	         \EndIf
	         \State $K \gets K+1$
	         \If{$(X_{K'}, X_{K'+1}, \ldots, X_{K})$ forms an alternating loop for some $1\leq K'< K$}
	         \State Flip every entry in the loop $(X_{K'}, X_{K'+1}, \ldots, X_K)$
	         \State \textbf{Break}
	         \EndIf

		\EndWhile

		\State\textbf{Return: The updated $M$} 
	\end{algorithmic}
\end{algorithm}

The column-one sampling set in Algorithm \ref{alg:D-Snake} needs no explicit diagonal exclusion because all label-diagonal cells in $D$ are structurally fixed at zero and therefore cannot be one-candidates.

The new algorithm has essentially the same idea as the original \snake algorithm. We enforce the diagonal constraints by forcing the loop never to touch the diagonal elements. The following proposition records the resulting validity statement; its proof is in Appendix \ref{app:proof-dsnake}.

\begin{prop}\label{prop:dsnake-correctness}
Assume $\Sigma_D(\row,\col;U,V)$ is nonempty after preprocessing. The transition kernel of Algorithm \ref{alg:D-Snake} is reversible with respect to the uniform distribution on $\Sigma_D(\row,\col;U,V)$. If the directed fixed-degree state graph is connected by directed $2$-switches and directed $3$-cycle reorientations, then Algorithm \ref{alg:D-Snake} is irreducible and has the uniform distribution on $\Sigma_D(\row,\col;U,V)$ as its unique stationary distribution.
\end{prop}

The connectivity condition in Proposition \ref{prop:dsnake-correctness} is the standard one for simple directed graphs with prescribed in- and out-degrees. In matrix language, \citet{rao1996markov} identify the needed moves as ordinary checkerboard switches together with compact alternating hexagons; \citet{lamar2009uniform} states the resulting Rao--Jana--Bandyopadhyay theorem as connectivity of the meta-graph whose edges are directed $2$-switches and directed $3$-cycle reorientations. Equivalently, \citet{berger2010uniform} construct the corresponding directed state graph and prove that it is strongly connected, while also noting that directed $2$-switches alone can fail to connect all realizations for some degree sequences. A general mixing-time analysis with structural diagonal constraints remains open.

\subsubsection{Equal-margin label shuffling}
When row or column margins have repeated values, the target distribution is invariant under permutations within equal-margin row classes and within equal-margin column classes. Let $Q$ denote the kernel that independently applies uniform random permutations inside all equal-margin row and column classes. Since $Q$ is a measure-preserving relabeling of $\Sigma(\row,\col)$ and $\bP_\snake$ preserves the uniform law, any composition of $Q$ and $\bP_\snake$ also preserves the uniform law. The $\snakep$ variant used in the numerical section is the block kernel that performs five ordinary \snake transitions and then applies one $Q$ relabeling. This block sampler is generally non-reversible, but the uniform law remains invariant. In Figure \ref{fig:mixing_sparse}, the implementation uses this five-attempt block schedule.

\section{Theoretical analysis}\label{sec:theoretical-analysis}
This section presents the theoretical analysis concerning the correctness and efficiency of the \snake algorithm. To start, we recall a few concepts in the theory of finite Markov chains. A Markov chain with transition kernel $\bP$ is said to be reversible with respect to some distribution $\pi$ if $\pi(x)P(x,y) = \pi(y)P(y,x)$ for any pair of states $(x,y)$. We say $\bP$ is $\pi$-irreducible if for any two states satisfying $\pi(x), \pi(y) > 0$, there exists some positive integer $N = N(x,y)$ such that $\bP^N(x,y) > 0$. Meanwhile, if the Markov chain $\bP$ is both  $\pi$-reversible and irreducible, then $\pi$ is the unique stationary distribution. 
Throughout this section, we use the preprocessing convention from Section \ref{sec: Snake algorithm}: deterministic rows and columns have been removed, so $1\le r_i\le n-1$ and $1\le c_j\le m-1$ for the active matrix.

\subsection{Correctness}
The following theorem shows our \snake algorithm has the correct stationary distribution.
\begin{theorem}\label{thm: correctness}
Suppose we have fixed row sums $\row = (r_1, r_2, \ldots, r_m) \in \bN^m$ and column sums $\col = (c_1, c_2, \ldots, c_n) \in \bN^n$ such that $\Sigma(\row, \col) \neq \varnothing$. Let $\bP_\snake$ be the Markov transition kernel described by Algorithm \ref{alg:Snake}. Then $\bP_\snake$ is irreducible and $U(\Sigma(\row, \col))$-reversible. 
\end{theorem}

Although the technical details are deferred to Appendix \ref{app:proofs}, we can sketch the proof ideas here. The irreducibility of our Algorithm essentially follows from the irreducibility of the classical swap algorithm. Let $M_1, M_2$ be two arbitrary matrices with the same row and column sums. It is known in Theorem 3.1 of \cite{ryser1957combinatorial} that $M_1$ is transformable into $M_2$ by a finite number of swaps. Since every swap can be realized with positive probability by the \snake algorithm, any two matrices in $\Sigma(r, c)$ are connected by a finite sequence of \snake transitions, which proves the irreducibility.

Since our target distribution is uniform over all the matrices in $\Sigma(\row,\col)$, showing the \snake algorithm is $U(\Sigma(\row,\col))$-reversible is equivalent to proving the transition kernel $\bP_\snake(\cdot, \cdot)$ is symmetric. Given two matrices $A_1, A_2$ that differ by exactly one loop (otherwise both $\bP_\snake(A_1,A_2) = \bP_\snake(A_2,A_1) = 0$, which is obviously symmetric), we may assume without loss of generality that we are currently at $A_1$.
There are multiple paths in one step of the \snake algorithm that transfer $A_1$ to $A_2$, and $\bP_\snake(A_1,A_2)$ is the summation of the probability of all such paths. We will define a bijection from all the paths which transfer $A_1$ to $A_2$ to those that transfer $A_2$ to $A_1$. We will also show the probability of choosing each path (given current state $A_1$) is the same as the probability of choosing the path in the image of the bijection (given current state $A_2$), and this in turn shows $\bP_\snake(A_1,A_2) = \bP_\snake(A_2,A_1)$.
\subsection{Move size and computational cost analysis}\label{subsec: moresize}
After showing the validity of our algorithm, we investigate the efficiency of the \snake algorithm in this subsection. We mainly consider two quantities: the number of entries flipped by a transition and the computational cost of generating that transition. Though closely related, the former is a move-size diagnostic that does not involve wall-clock time, whereas the latter concerns the actual algorithmic cost. 
\subsubsection{Move size: Number of flipped entries}\label{subsubsec:number of flipped entries}
The number of flipped entries measures the algorithm's step size on the space of binary matrices. A larger step size means that a transition changes a larger part of the current matrix, but it does not by itself imply a shorter total-variation mixing time or a smaller Monte Carlo standard error for every statistic. We therefore use this quantity as a structural move-size diagnostic for comparing loop-based and local-flip algorithms. Both the swap and the rectangle loop algorithm flip at most four entries, as at most, only a $2\times 2$ submatrix is swapped in each iteration. When the row and column sums are sparse, the expected number of flipped entries for both the swap and rectangle loop algorithm goes to zero when the dimensionality goes to infinity. Therefore these classical local algorithms can be prohibitively slow, especially when the column and row sums are sparse and the dimension is high. In contrast, one step of the \snake algorithm flips on the order of $\sqrt {n}$ entries under relatively general assumptions, giving substantially more nonlocal moves than these local methods. The numerical section also compares against Curveball, a stronger row-trade method that is not covered by this local-flip discussion. The following two results estimate the expected number of flipped entries per step when the matrices are sparse or relatively regular. 
Throughout Theorems \ref{thm: flips-sparse}--\ref{thm: rate optimality} and Theorem \ref{thm:half-balanced-efficiency}, expectations are conditional on an arbitrary current state $M\in\Sigma(\row,\col)$ and are taken only over the internal randomness of one \snake transition; the stated bounds are uniform over the choice of $M$.

\begin{theorem}[Sparse]\label{thm: flips-sparse}
Let $L$ be the number of flipped entries in one step of the \snake algorithm. Then 
 $$
    \bE(L)\geq\frac{1}{64}
    \sqrt{\max\left\{
    {\frac{n}{\rmax}},
        {\frac{m}{m-\cmin}}
    \right\}},$$ where $\rmax$ and $\cmin$ are the largest row sum and smallest column sum defined in Section \ref{sec: Snake algorithm}.
\end{theorem}

Theorem~\ref{thm: flips-sparse} shows, for example, that if $m\asymp n$ and either
$\rmax=\calO(1)$ or $m-\cmin=\calO(1)$, then the expected number of flipped entries is
$\Omega(\sqrt n)$. This compares favorably with classical local-flip algorithms.

Another important case is when active columns contain a fixed fraction of ones and active rows contain a fixed fraction of zeros. The next theorem shows the algorithm flips the same magnitude of entries under this setting. 

\begin{theorem}[Balanced]\label{thm: flips-balance}
 Let $\delta$ and $\Delta$ be constants with $0 < \delta\leq \Delta < 1$. Assume that  $m \asymp n, ~\cmin \geq \delta m, \text{and} ~ \rmax \leq \Delta n$. Then we have 
    $\bE(L) = \Theta(\sqrt{n})$.
\end{theorem}

\subsubsection{Computational efficiency}
Though the \snake algorithm flips many entries per step, the  cost may also be higher than existing methods. To make a fair comparison of the computational efficiency, we now turn to study the computational complexity of our algorithm. To run one iteration of the \snake algorithm, one needs to run the inner while loop for a random number of iterations, denoted by $C$. If we refer to the cost of one iteration in the while loop by one unit, then the computational cost of the \snake algorithm is $C$ per step. The following result gives a universal upper bound on $C$ without assumption on the margin sums. 
\begin{theorem}\label{thm: upper bound}
 We have $\bE(C) = \calO(\log(m+n)\sqrt{m+n})$.
\end{theorem}
In one step of the \snake algorithm, since we only flip those entries that form an alternating loop (see Figure \ref{fig:snake} as an example), we have $L=O(C)$. The ratio $\bE(C)/\bE(L)$ measures the expected work per flipped entry. Given margin sums $\row$ and $\col$, we say the computational efficiency of the \snake algorithm is rate optimal if $\bE(C)/\bE(L) = \Theta(1)$. Our next result shows the \snake algorithm is rate optimal in the sparse and balanced regimes above.

\begin{theorem}\label{thm: rate optimality}
 With all the notations defined  above, we have:
 \begin{itemize}
    \item (Sparse) If $\rmax = \cmax = \calO(1)$, then  the \snake algorithm is rate optimal. 
     \item (Balanced) If $\row$ and $\col$ satisfy the assumption in Theorem \ref{thm: flips-balance}, then 
   the \snake algorithm is rate optimal. 
 \end{itemize}
\end{theorem}

The next theorem gives a weaker but more one-sided guarantee. It does not require $m\asymp n$ and does not assert constant work per flipped entry; it shows that a half-balanced margin condition is enough to prevent the closing loop from being shorter than the explored path by more than a polylogarithmic factor in expectation. Since $L=O(C)$ for every transition, any sampler of this form has $\bE(C)/\bE(L)=\Omega(1)$; the theorem is therefore near-optimal up to the displayed polylogarithmic factor.

\begin{theorem}[Half-balanced efficiency]\label{thm:half-balanced-efficiency}\leavevmode\par\noindent
Let $\delta$ and $\Delta$ be fixed constants with $0<\delta<1/2<\Delta<1$. Assume that either
\[
    \rmax\le \delta n
    \qquad\text{or}\qquad
    \cmin\ge \Delta m .
\]
Then one step of the \snake algorithm satisfies
\[
    \frac{\bE(C)}{\bE(L)}=\calO_{\delta,\Delta}\bigl(\log^{13}(m+n)\bigr).
\]
\end{theorem}
The efficiency of our algorithm also compares favorably to existing algorithms, especially in the sparse case. For example, assuming $m = n$ and $\rmax = \cmax = \calO(1)$, the swap algorithm has $\bE(C) = \Theta(1)$ per attempted swap and $\bE(L) = \calO(1/n^2)$ flipped entries per attempt. Thus its work per flipped entry, measured by $\bE(C)/\bE(L)$, is $\Omega(n^2)$; equivalently, its flipped entries per unit work are only $\calO(1/n^2)$. 

\subsection{Mixing time}
Theorem \ref{thm: correctness} gives irreducibility and reversibility of the raw \snake chain. The permutation-matrix theorem below analyzes the raw chain directly, including the parity issue through the random cycle-length distribution. For the general fixed-margin comparison theorem, we use lazy versions of the chains to remove possible periodicity in finite state spaces. Fix margin sums $\row, \col$ and define $\bP_{\snake,\mathrm{lazy}}=(I+\bP_{\snake})/2$. The $\epsilon$-mixing time of the lazy \snake algorithm is 
$$t_{\mix}^\snake(\epsilon):= \inf_t \{\sup_{M_0\in\Sigma(\row,\col)}\lVert \bP_{\snake,\mathrm{lazy}}^t(M_0,\cdot) - U(\Sigma(\row,\col)) \rVert_\TV \leq \epsilon\},$$
where $\lVert \mu - \nu \rVert_\TV = \sup_A \lvert \mu(A) - \nu(A) \rvert$ is the total-variation (TV) distance between two probability measures. Defining $\bP_{\swap,\mathrm{lazy}}=(I+\bP_{\swap})/2$ gives $t_{\mix}^\swap(\epsilon)$ analogously for the lazy swap algorithm. For a reversible kernel $P$, write $\gap(P)=1-\lambda_2(P)$ for its spectral gap. We say a Markov chain on the space of $m\times n$ binary matrices is rapidly mixing (or polynomial mixing) if the $\epsilon$-mixing time is upper bounded by a polynomial of $(m,n,1/\epsilon)$. Establishing such bounds is difficult because $|\Sigma(\row,\col)|$ is often exponential in $m,n$. \citet{kannan1999simple} conjectured that the swap algorithm is rapidly mixing for every feasible pair of margins. Earlier work established rapid mixing for many structured families; see \citet{erdHos2022mixing} for a unified treatment, and \citet{tikhomirov2020sharp,tikhomirov2022regularized} for sharper functional inequalities and mixing estimates in regular regimes. The conjecture for arbitrary feasible margins was recently resolved by \citet{fu2026spectral}, who prove the worst-case-sharp bound
\[
    \gap(\bP_{\swap,\mathrm{lazy}})
    \geq \binom{m}{2}^{-1}\binom{n}{2}^{-1}.
\]
For a singleton fixed-margin state space, we adopt the convention $\gap(P)=1$; all mixing statements are then immediate.

Before using this general swap-chain result, we give a sharper direct analysis in the sparsest non-degenerate square case. When $\row=\col=\mathbf{1}_n$, each matrix in $\Sigma(\mathbf{1}_n,\mathbf{1}_n)$ is a permutation matrix, so the state space can be identified with $S_n$. In this case, one step of the \snake algorithm is a random cycle walk whose cycle length has the birthday scale $\Theta(\sqrt n)$.
\begin{theorem}[Permutation matrices]\label{thm:perm-mixing}
Assume $m=n$ and $\row=\col=\mathbf{1}_n$. Let $t_{\mathrm{mix}}^{\mathrm{perm}}(\epsilon)$ denote the total-variation mixing time of the raw \snake chain on $\Sigma(\mathbf{1}_n,\mathbf{1}_n)$. Then
\[
    t_{\mathrm{mix}}^{\mathrm{perm}}(1/4)=\Theta(\sqrt n\log n).
\]
Consequently, for every fixed $0<\epsilon<1/2$,
\[
    t_{\mathrm{mix}}^{\mathrm{perm}}(\epsilon)=\Theta_\epsilon(\sqrt n\log n).
\]
\end{theorem}

The lower bound in Theorem \ref{thm:perm-mixing} follows from a coupon-collector obstruction: before enough raw labels have been touched, many labels are still fixed. The upper bound uses Fourier analysis on $S_n$. After conditioning on the event that the snake cycle has length in a constant-width interval around $\sqrt n$, Hough's character-ratio estimates for random $k$-cycles imply contraction in all non-trivial, non-sign irreducible representations; the sign representation is controlled by the smoothing of the random cycle length distribution. The full proof is in Appendix \ref{app:proof-perm-mixing}.

For general fixed margins, combining this new swap-chain bound with a direct Dirichlet-form comparison gives a universal polynomial bound for \textsc{Snake}. The explicit estimate is conservative---it only uses the probability that a \snake trajectory realizes one prescribed checkerboard swap---but it requires no sparsity, density, regularity, or stability assumption on the margins.
\begin{theorem}[All feasible margins]\label{thm: rapid mixing}
Let $\row,\col$ be any feasible pair of margins for an active $m\times n$ matrix, and put $d=\max\{m,n\}$. Then
\[
    \gap(\bP_{\snake,\mathrm{lazy}})\geq \frac{1}{mn d^3}.
\]
Consequently, for every $0<\epsilon<1$,
\[
    t_{\mix}^{\snake}(\epsilon)
    \leq mn d^3\left(\log |\Sigma(\row,\col)|+\log\epsilon^{-1}\right)
    \leq mn d^3\left(mn\log 2+\log\epsilon^{-1}\right).
\]
In particular, the lazy \snake chain is rapidly mixing for every feasible pair of margins.
\end{theorem}

The universal estimate in Theorem \ref{thm: rapid mixing} settles polynomial mixing but is not expected to be sharp for \textsc{Snake}. Theorem \ref{thm:perm-mixing} shows how much stronger a margin-specific analysis can be: in the canonical permutation-matrix case it gives the sharp order $\Theta(\sqrt n\log n)$ for the raw chain. Obtaining sharp \snake mixing-time estimates for broader margin classes remains an important open problem.

\section{Numerical Experiments}\label{sec:numerical}
In this section, we examine the empirical performance of the \snake algorithm and compare it with existing algorithms. The first set of experiments illustrates the move-size results from Section~\ref{subsubsec:number of flipped entries}. Then, we compare empirical convergence diagnostics among existing algorithms via commonly used test statistics. We compare the \snake algorithm with conditional-Poisson sequential importance sampling (SIS) on both the finch null-model example and two exactly calibrated Rasch-model examples. Finally, we implement the diagonal-constrained \dsnake variant and compare it with a classical directed edge-swap algorithm on a fixed-degree directed-network null problem. All timings in this section use compiled Rcpp implementations on the same machine. 
\subsection{Number of flipped entries}
Besides being rejection-free, an important theoretical advantage of the \snake algorithm is that it flips a lot more entries per step than existing algorithms such as the swap and rectangle loop algorithm.   
When $m \asymp n$, the results in Section~\ref{subsec: moresize} show that the \snake algorithm flips $\Theta(\sqrt{n})$ entries per step when the margin sums are either sparse or relatively balanced. To numerically investigate this, we design two experiments. In the first experiment, we fix the matrix size and vary the filled proportion. In the second experiment, we consider the sparsest possible case ($\row = \col = (1,1,\ldots, 1)^\top)$ and vary the dimension. We compare the number of flipped entries per iteration as a move-size diagnostic and the number of flipped entries per second as an implementation-throughput diagnostic. In each setting, each algorithm is implemented for $10^5$ steps.

In the first experiment, we choose $m = n = 300$ so the matrix is moderately large. For each $p \in \{0.01,0.05,0.1,0.2,0.3,0.4,0.5\}$, we initialize the matrices by generating independent $\Bern(p)$ random variables for every entry and remove deterministic rows and columns before sampling. For fixed fill probabilities bounded away from $0$ and $1$, such Bernoulli-initialized matrices satisfy balanced-type margin bounds with high probability as the dimension grows; the lowest-fill cases in Table \ref{tab:comparison_fix_size} should instead be read as finite sparse examples. Our results are summarized in Table \ref{tab:comparison_fix_size}. The \snake algorithm flips about $18$ entries per iteration across all fill levels. In contrast, Rectangle Loop and Swap become much more local as the matrix becomes sparse. The optimized \snake implementation also has the largest flipped-entries-per-second rate in each setting in Table \ref{tab:comparison_fix_size}.

We also consider the sparse case which satisfies the assumptions in Theorem \ref{thm: flips-sparse}. We consider the sparsest case where $\row = \col = (1,1,\ldots, 1)^\top \in \mathbb N^{n}$, and choose $n$ in $\{100,150,\ldots, 1000\}$ to understand our algorithm's scalability. Our results are presented in Figure \ref{fig:flipped_entries_sparse}. In this ultrasparse case, both the swap and the Rectangle Loop algorithm work poorly when dimension grows. Across five independent runs with $10^5$ attempts per run, the swap algorithm averages no more than $8.8$ total flipped entries for any $n\geq 300$, and at most $8\times 10^{-6}$ flipped entries per attempt for $n\geq 700$. The rectangle loop flips at most $0.0198\pm 0.0004$ entries per attempt when $n\geq 400$, with the maximum in that range occurring at $n=400$, and its move size decays as $n$ increases. In contrast, the mean number of flipped entries per attempt for the \snake algorithm increases from $15.122\pm 0.013$ at $n=100$ to $42.272\pm 0.056$ at $n=1000$, consistent with the $\sqrt n$ scale in Theorem~\ref{thm: flips-sparse} and with the more precise permutation-cycle calculation in Lemma~\ref{lem:perm-birthday}. Thus, \snake maintains much larger moves than these local methods when the matrices are very sparse.  

\begin{table}[htbp]
\centering
\small
\begin{tabular}{c c c l r r}
\toprule
Target fill & Active size & Actual fill & Method & Flips/iteration & Flips/second \\
\midrule
\multirow{3}{*}{$1\%$} & \multirow{3}{*}{$291\times 284$} & \multirow{3}{*}{$0.0106$}
& \snake & $\mathbf{18.0315}$ & $\mathbf{2.10\times 10^7}$ \\
& & & Rectangle Loop & $0.0821$ & $1.37\times 10^6$ \\
& & & Swap & $0.0006$ & $2.13\times 10^4$ \\
\midrule
\multirow{3}{*}{$5\%$} & \multirow{3}{*}{$300\times 300$} & \multirow{3}{*}{$0.0507$}
& \snake & $\mathbf{18.1475}$ & $\mathbf{2.11\times 10^7}$ \\
& & & Rectangle Loop & $0.3784$ & $7.57\times 10^6$ \\
& & & Swap & $0.0191$ & $6.37\times 10^5$ \\
\midrule
\multirow{3}{*}{$10\%$} & \multirow{3}{*}{$300\times 300$} & \multirow{3}{*}{$0.0997$}
& \snake & $\mathbf{18.1524}$ & $\mathbf{2.06\times 10^7}$ \\
& & & Rectangle Loop & $0.7167$ & $1.19\times 10^7$ \\
& & & Swap & $0.0636$ & $2.12\times 10^6$ \\
\midrule
\multirow{3}{*}{$20\%$} & \multirow{3}{*}{$300\times 300$} & \multirow{3}{*}{$0.1991$}
& \snake & $\mathbf{18.1686}$ & $\mathbf{2.04\times 10^7}$ \\
& & & Rectangle Loop & $1.2815$ & $1.60\times 10^7$ \\
& & & Swap & $0.2007$ & $5.02\times 10^6$ \\
\midrule
\multirow{3}{*}{$30\%$} & \multirow{3}{*}{$300\times 300$} & \multirow{3}{*}{$0.3005$}
& \snake & $\mathbf{18.2089}$ & $\mathbf{2.05\times 10^7}$ \\
& & & Rectangle Loop & $1.6710$ & $1.86\times 10^7$ \\
& & & Swap & $0.3577$ & $8.94\times 10^6$ \\
\midrule
\multirow{3}{*}{$40\%$} & \multirow{3}{*}{$300\times 300$} & \multirow{3}{*}{$0.4010$}
& \snake & $\mathbf{18.1693}$ & $\mathbf{2.04\times 10^7}$ \\
& & & Rectangle Loop & $1.9233$ & $1.92\times 10^7$ \\
& & & Swap & $0.4608$ & $9.22\times 10^6$ \\
\midrule
\multirow{3}{*}{$50\%$} & \multirow{3}{*}{$300\times 300$} & \multirow{3}{*}{$0.4989$}
& \snake & $\mathbf{18.0941}$ & $\mathbf{1.95\times 10^7}$ \\
& & & Rectangle Loop & $1.9868$ & $1.81\times 10^7$ \\
& & & Swap & $0.5030$ & $1.01\times 10^7$ \\
\bottomrule
\end{tabular}
\caption{Comparison between the \snake, Rectangle Loop, and Swap algorithms on Bernoulli-initialized $300\times 300$ matrices. Rows and columns with deterministic margins are removed before running the chains; the active matrix dimensions and actual fill rates are shown. Each entry is based on $10^5$ iterations.}
\label{tab:comparison_fix_size}
\end{table}

\begin{figure}[htbp!]
\includegraphics[width = \textwidth]{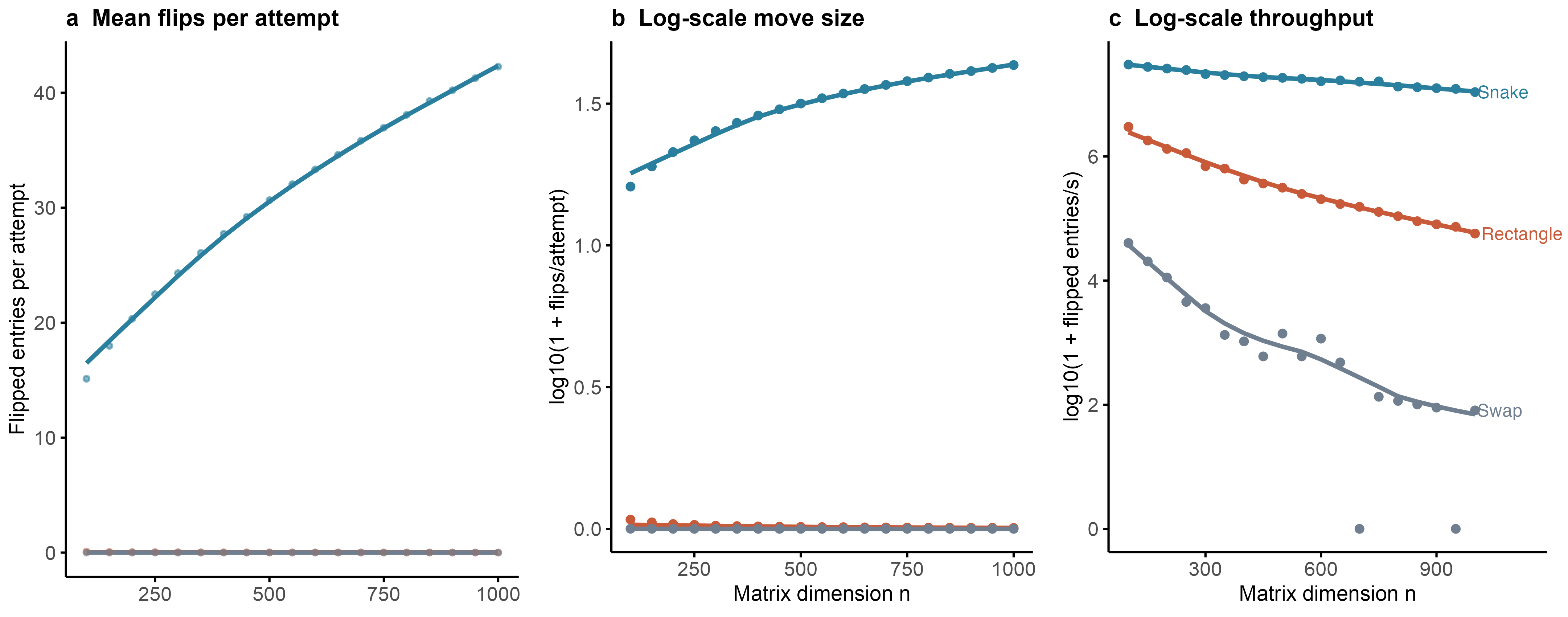}
\caption{Flipped entries in the sparsest square case $\row = \col = (1,1,\ldots, 1)^\top$, using five independent runs with $10^5$ attempts per method and dimension. Panel a gives the mean flipped entries per attempt; the shaded band is one run standard error. Panel b shows $\log_{10}(1+\text{flips/attempt})$ to separate the local methods near zero. Panel c shows implementation throughput as $\log_{10}(1+\text{flipped entries/second})$.}
\label{fig:flipped_entries_sparse}
\end{figure}

\subsection{Empirical convergence diagnostics on large, sparse matrices}
In practice, most large binary matrices are sparse. Therefore we take a closer look at the performance of our method on large, sparse matrices. To empirically compare convergence diagnostics among existing methods, we use a perturbation score $T_S$, a squared row-pair overlap score $T_1$ related to statistics used in \cite{ponocny2001nonparametric, verhelst2008efficient}, and a local adjacent-column overlap score $T_2$.

Given a matrix $M_1$ with the same margins as the initialization matrix $M_0$, let
\[
N_1:=\sum_{i=1}^m\sum_{j=1}^n M_0(i,j),
\]
which is also the number of ones in $M_1$. Following \citet{strona2014fast}, we define the perturbation score by
\[
T_S(M_1;M_0)
:=
\frac{\sum_{i=1}^m\sum_{j=1}^n
M_0(i,j)\bigl(1-M_1(i,j)\bigr)}
{N_1}.
\]
Thus, $T_S(M_1;M_0)$ is the fraction of one-entries in $M_0$ that are no longer occupied in $M_1$. Since $M_0$ and $M_1$ have the same number of ones, it can equivalently be written as the normalized Hamming distance
\[
T_S(M_1;M_0)
=
\frac{1}{2N_1}
\sum_{i=1}^m\sum_{j=1}^n
\left|M_1(i,j)-M_0(i,j)\right|.
\]
The row-pair score is
\[
T_1(M_1)=\binom{m}{2}^{-1}\sum_{1\le i<j\le m}\bigl((M_1M_1^\top)_{ij}\bigr)^2,
\]
the average squared overlap over all unordered pairs of distinct rows. In the square banded examples, the adjacent-column diagnostic is
\[
T_2(M_1)=n^{-1}\sum_{j=1}^n\sum_{i=1}^n (M_1)_{ij}(M_1)_{i,j+1},
\]
with the column index interpreted cyclically. This last statistic is order-sensitive and measures how quickly local column structure from the banded initialization is destroyed.

We choose five algorithms: the \snake algorithm, the $\snakep$ algorithm, the swap algorithm \cite{besag1989generalized}, the rectangle loop algorithm \cite{wang2020fast}, and the Curveball algorithm \cite{strona2014fast} for comparison. The initial matrices are $n\times n$ banded matrices with row and column sums equal to $10$. Figure \ref{fig:mixing_sparse} focuses on the replicated $n=1000$ setting, where diagnostics are plotted against elapsed sampling time on a logarithmic scale; $\snakep$ performs an equal-margin label shuffle every five \snake transition attempts. Because Swap and Rectangle Loop are much cheaper per attempt than the nonlocal methods, each method is run for enough transition attempts to cover a comparable wall-clock window rather than for an identical attempt count. Curveball is the strongest practical comparator in this experiment because it also makes nonlocal row trades. On this elapsed-time scale, $\snakep$ reaches the perturbation threshold much faster than the other methods, \snake reaches the displayed squared row-pair-overlap threshold faster than Curveball, and $\snakep$ reaches the adjacent-column-overlap threshold fastest. Swap does not reach any of the three displayed panel-a--c thresholds in this window.

At the displayed panel-a--c thresholds, the mean elapsed sampling times across five runs are as follows. For $T_S\ge 0.95$, $\snakep$ takes $0.0008\pm0.0001$ seconds, Curveball $0.0300\pm0.0011$ seconds, \snake $0.0413\pm0.0025$ seconds, and Rectangle Loop $0.1433\pm0.0172$ seconds. For $T_1\le0.10$, \snake takes $0.0153\pm0.0010$ seconds, $\snakep$ $0.0180\pm0.0012$ seconds, Curveball $0.0218\pm0.0011$ seconds, and Rectangle Loop $0.0505\pm0.0041$ seconds; Figure \ref{fig:mixing_sparse}d displays this comparison directly. For $T_2\le0.10$, $\snakep$ takes $0.0011\pm0.0002$ seconds, \snake $0.0456\pm0.0033$ seconds, Curveball $0.0643\pm0.0017$ seconds, and Rectangle Loop $0.1486\pm0.0204$ seconds.

\begin{figure}[htbp!]
\includegraphics[width = \textwidth]{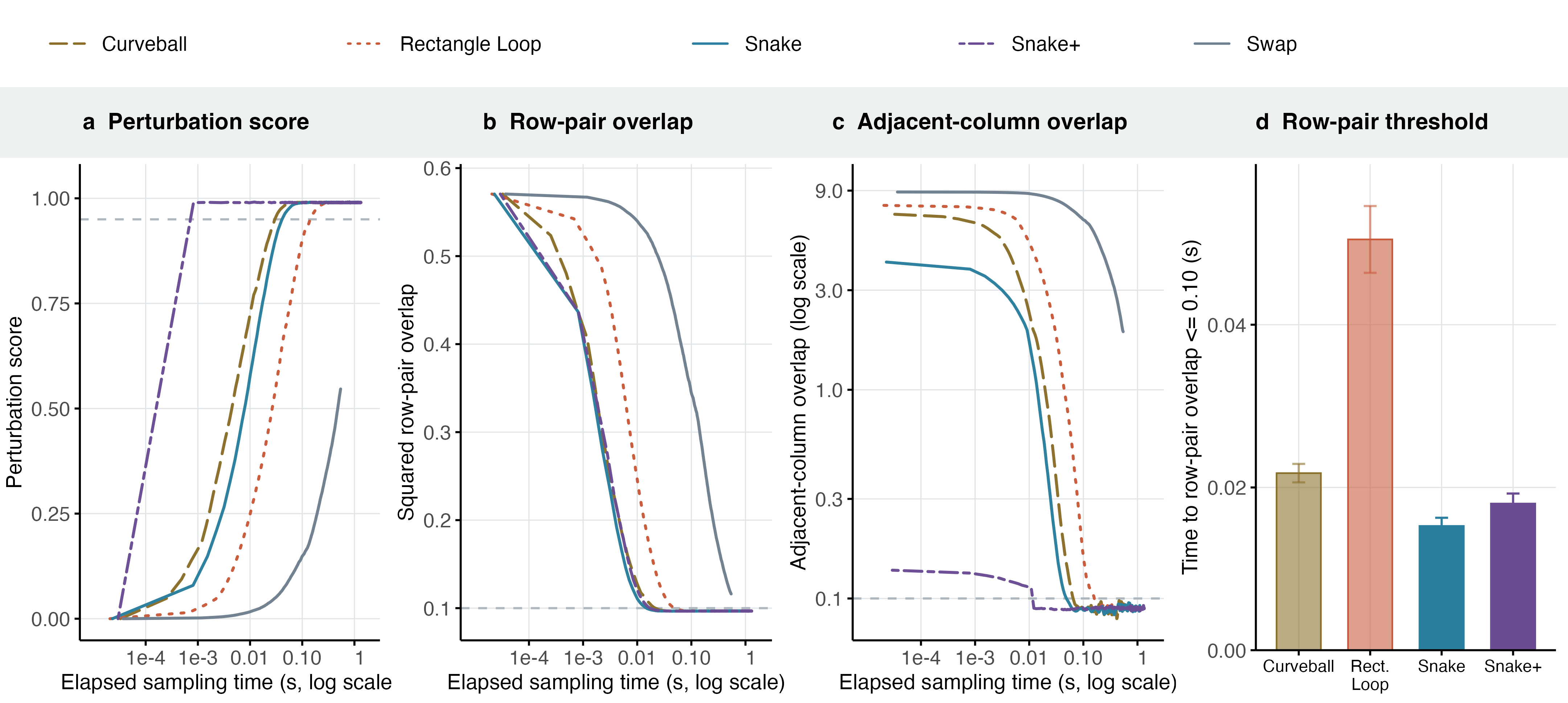}
\caption{Replicated convergence diagnostics on sparse $1000\times 1000$ banded matrices with row and column sums equal to $10$. Panels a and b show unsmoothed run-mean traces across five independent runs; panel c shows the corresponding smoothed run-mean trace for the adjacent-column diagnostic. Each method is drawn once. The horizontal axis in panels a--c is elapsed sampling time in seconds on a log scale. Dashed horizontal lines mark the displayed thresholds: perturbation score at least $0.95$ in panel a, squared row-pair overlap at most $0.10$ in panel b, and adjacent-column overlap at most $0.10$ in panel c. Panel c uses a logarithmic vertical scale. Panel d shows the mean elapsed sampling time to reach squared row-pair overlap at most $0.10$ for the four methods that reach the threshold; error bars are one run standard error. Swap is omitted from panel d because it does not reach this threshold in any run.}
\label{fig:mixing_sparse}
\end{figure}

\subsection{Comparison with sequential importance sampling}
Sequential Importance Sampling (SIS) is also routinely used to approximate statistics with respect to an intractable target distribution. In contrast to MCMC algorithms, SIS methods typically generate independent weighted samples whose proposal distribution is not the target distribution. In the context of binary matrices with fixed margins, a popular SIS algorithm is developed by \cite{chen2005sequential}. When the row and column sums are sparse, theoretical properties of this algorithm have been investigated by \cite{blanchet2009efficient}. We first compare \snake with a conditional-Poisson SIS implementation on the finch example. We then compare both methods with exact hypergeometric tail probabilities in two dense, regular Rasch examples, where the exact values provide ground truth and expose the computational cost of generating an entire matrix sequentially.

\subsubsection{Monte Carlo estimands and uncertainty}
For a statistic $h:\Sigma(\row,\col)\to\mathbb R$, the target mean is
\[
    \theta_h=\mathbb E_U(h(M)),
\]
where $U=U(\Sigma(\row,\col))$. For an observed value $x$, we write the strict and inclusive upper-tail probabilities as
\[
    p_h^{>}(x)=\mathbb P_U(h(M)>x), \qquad
    p_h^{\ge}(x)=\mathbb P_U(h(M)\ge x).
\]
The reported \snake estimates are time averages from the raw rejection-free chain. The raw chain has the uniform stationary distribution by Theorem \ref{thm: correctness}; for finite irreducible chains, such ergodic averages are consistent for stationary expectations, while the lazy chain is used in Section \ref{sec:theoretical-analysis} for the general fixed-margin total-variation comparison statement. Given draws $M_1,\ldots,M_N$ and burn-in $B<N$, with $K=N-B$, we estimate
\[
    \widehat\theta_h=\frac{1}{K}\sum_{t=B+1}^{N} h(M_t), \qquad
    \widehat p_h^{>}(x)=\frac{1}{K}\sum_{t=B+1}^{N}\indc{h(M_t)>x},
\]
and analogously for $\widehat p_h^{\ge}(x)$. All MCMC standard errors reported below are batch standard errors: the post-burn-in scalar sequence is divided into contiguous batches of size $1000$, any incomplete final batch is discarded, and the reported standard error is the sample standard deviation of the batch means divided by the square root of the number of batches. For the finch example, this leaves $400{,}000$ post-burn-in samples and about $147$ raw strict-tail hits, so the uncertainty in this small tail probability is not negligible.

For SIS, let $q_s$ be the sequential proposal probability of the $s$th sampled matrix and define the normalized importance weights
\[
    \bar w_s=\frac{q_s^{-1}}{\sum_{\ell=1}^N q_\ell^{-1}}.
\]
The corresponding inclusive-tail estimator is
\[
    \widehat p_{h,\mathrm{SIS}}^{\ge}(x)
    =\sum_{s=1}^N\bar w_s\indc{h(M_s)\ge x},
\]
with effective sample size $\mathrm{ESS}=(\sum_s\bar w_s^2)^{-1}$. We report the independent importance-sampling plug-in standard error
\[
    \widehat{\mathrm{se}}_{\mathrm{SIS}}
    =\left(\sum_{s=1}^N\bar w_s^2
    \left(\indc{h(M_s)\ge x}-\widehat p_{h,\mathrm{SIS}}^{\ge}(x)\right)^2\right)^{1/2}.
\]

\subsubsection{Testing null models}
We first consider Darwin's finch data collected during Darwin's visit to Gal\'apagos. The data represent the occurrence or absence of $13$ species of finches on $17$ islands. Given an observed matrix $M(0)$ with margin sums $\row,\col$, we use the standard statistic
\[
S^2(M) := \frac{1}{m(m-1)}\sum_{i\neq j} (M M^\top)_{i,j}^2
\]
suggested by \cite{roberts1990island} to test whether $M(0)$ is uniformly distributed over $\Sigma(\row, \col)$. The $S^2$ tail probabilities below use strict upper tails. The observed data have $S^2(M(0)) = 53.1154$. After removing one deterministic row for sampling and restoring it for scoring, we run the \snake algorithm for $5\times 10^5$ steps and collect samples after a $10^5$-step burn-in period. The run takes $4.559$ seconds. The estimated mean of $S^2$ under $U(\Sigma(\row,\col))$ is $50.6968$, with batch standard error $0.0045$. The estimated tail probability $\bP(S^2(M)>S^2(M(0)))$ is $3.675\times 10^{-4}$, with batch standard error $6.45\times 10^{-5}$.

The conditional-Poisson SIS implementation generates $5\times 10^5$ weighted samples in $252.133$ seconds. Its weighted mean estimate is $50.7071$, the weighted tail-probability estimate is $2.43\times 10^{-4}$, and the effective sample size is $99669.35$. Both methods therefore give similar estimates for $\bE(S^2)$ and indicate strong evidence against the fixed-margin null model for the observed finch data. The left panel of Figure \ref{fig:histogram} shows the \snake MCMC sample distribution and the observed statistic.

Next we consider a $250\times 250$ matrix where each entry is independently generated from $\Bern(0.7)$. The generated $M(0)$ is uniformly distributed over the set of matrices with its realized margins. The observed statistic is $S^2=14851.86$. Running the \snake algorithm for $5\times 10^5$ steps takes $12.865$ seconds. The estimated mean of $S^2$ is $14851.50$, with batch standard error $0.0098$, and the estimated tail probability above the observed statistic is $0.0783$, with batch standard error $0.0084$. This non-extreme tail probability is consistent with the fact that the matrix was generated from the null model. The right panel of Figure \ref{fig:histogram} shows the corresponding Bernoulli null distribution.

\begin{figure}[H]
    \centering
    \includegraphics[width = \textwidth]{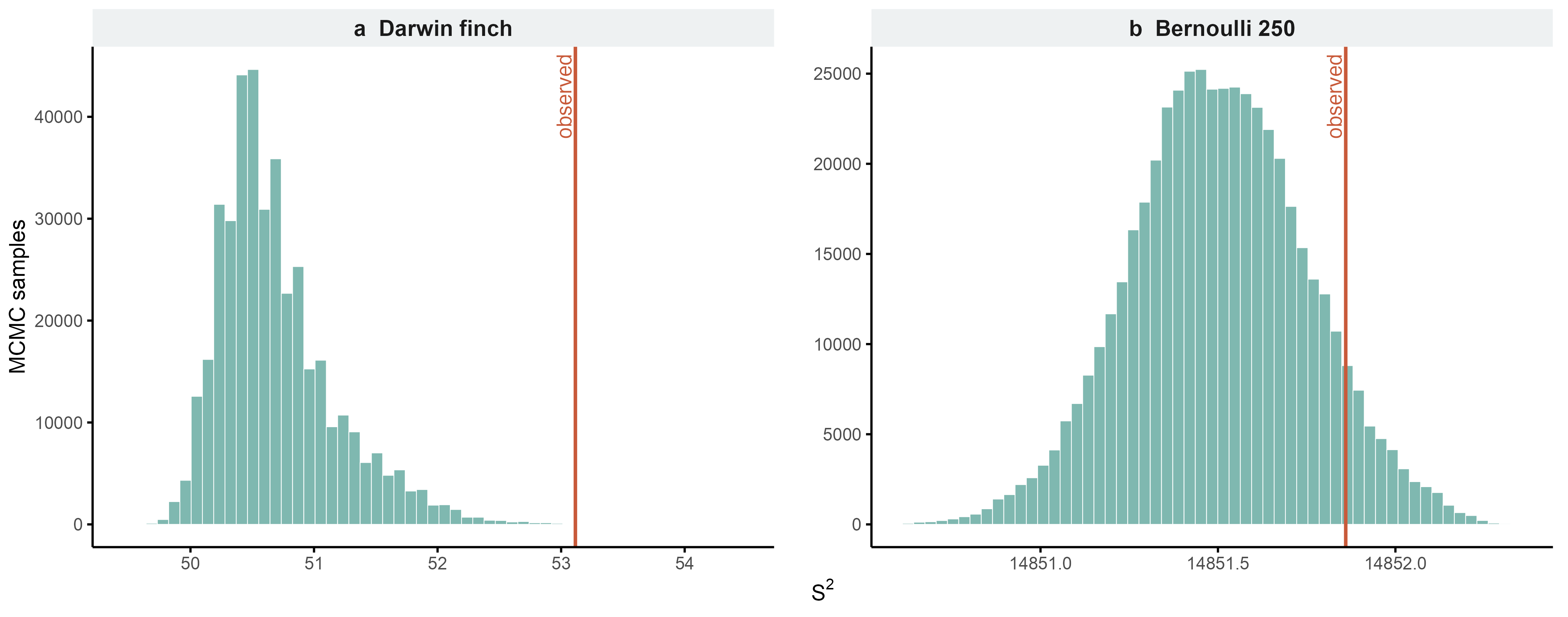}
    \caption{Histograms of the $S^2$ statistic from two \snake MCMC null-model runs. Left: Darwin finch data after a $10^5$-step burn-in, with the observed value $S^2(M(0))=53.1154$ in the far right tail. Right: a Bernoulli-initialized $250\times 250$ matrix with observed value $S^2(M(0))=14851.86$. Vertical lines mark the observed statistics.}
\label{fig:histogram}
\end{figure}

\subsubsection{Testing the Rasch model}
Here we use examples analyzed in \cite{ponocny2001nonparametric} and \cite{chen2005exact} to test the Rasch model described in the introduction. The Rasch examples use the inclusive upper-tail convention matching the exact hypergeometric calculation. In the first example, the binary matrix is of size $200\times 200$, representing the results of $200$ students answering $200$ questions. All student scores and item scores equal $100$. The students are divided into two groups of size $100$. The test statistic $f$ is the number of students in the first group who answer the first problem correctly. We estimate the tail probability for an observation with $f(M(0))=56$, namely $\bP(f(M)\geq 56)$. Since all row and column sums are the same, $f(M)$ has a hypergeometric distribution under the Rasch model; the exact tail probability is $0.0597903$.

We also consider a larger matrix of size $1000\times 1000$, where each row and each column has sum $500$. Again $f$ is the number of ones in the first half of the first column. For an observation with $f(M(0))=256$, the exact hypergeometric tail probability is $0.243318$. Thus both examples admit exact ground truth against which the \snake and SIS approximations can be checked.

The feasibility-filtered subset enumeration used by our small finch SIS implementation does not scale to these dense examples: a typical row proposal would have to consider subsets on the scale of $\binom{n}{n/2}$. We therefore implement a scalable conditional-Poisson row proposal through an elementary-symmetric-polynomial dynamic program. If $R$ rows remain and the residual sum of column $j$ is $c_j^{\mathrm{res}}$, the proposal odds for that column are $c_j^{\mathrm{res}}/(R-c_j^{\mathrm{res}})$. Computing the fixed-size normalizing constants and sampling a row with $k$ ones costs $O(nk)$, so constructing an entire $n\times n$ matrix costs $O(n^2k)$, or $O(n^3)$ when $k=n/2$. This avoids combinatorial enumeration but remains expensive because SIS constructs every entry of a new matrix for every independent draw. By contrast, a \snake transition only follows and updates one alternating path in the current matrix. The SIS implementation restarts an entire path if it reaches an infeasible residual state; no restart occurred in the reported regular dense runs.

For each size, we run \snake for $5.01\times10^7$ transitions, discard the first $10^5$, and retain $5\times10^7$ values of $f$. We use $5000$ SIS draws for the $200\times200$ case and $1000$ for the $1000\times1000$ case. Table \ref{tab:rasch-sis} gives the resulting three-way comparison.

\begin{table}[htbp]
\centering
\footnotesize
\setlength{\tabcolsep}{4pt}
\begin{tabular}{r r l r r r r r}
\toprule
$n$ & Exact $p$ & Method & Samples & Estimate & SE & ESS & Time (s) \\
\midrule
\multirow{2}{*}{$200$} & \multirow{2}{*}{$0.0597903$}
& \snake & $50{,}000{,}000$ & $0.0587820$ & $0.0016583$ & -- & $\mathbf{39.695}$ \\
& & Conditional-Poisson SIS & $5{,}000$ & $0.0613434$ & $0.0034303$ & $4927.0$ & $63.761$ \\
\midrule
\multirow{2}{*}{$1000$} & \multirow{2}{*}{$0.243318$}
& \snake & $50{,}000{,}000$ & $0.242400$ & $0.0118105$ & -- & $\mathbf{173.748}$ \\
& & Conditional-Poisson SIS & $1{,}000$ & $0.252201$ & $0.0137870$ & $995.2$ & $662.902$ \\
\bottomrule
\end{tabular}
\caption{Exact and Monte Carlo inclusive upper-tail probabilities in the two regular Rasch examples. To account conservatively for long-range autocorrelation, \snake standard errors use ten contiguous batches of size $5\times10^6$ after burn-in; SIS standard errors and effective sample sizes use the normalized importance weights. }
\label{tab:rasch-sis}
\end{table}

Both estimators agree with the exact tail probabilities to within their reported Monte Carlo uncertainty, so the comparison reduces to how much computation each needs to reach a given accuracy. We quantify the efficiency of each method using the variance--time product $\mathrm{SE}^2\times\text{time}$, which is smaller for the more efficient sampler. By this criterion, \snake is $6.9$ times as efficient for $n=200$ and $5.2$ times as efficient for $n=1000$. The advantage arises because the conditional-Poisson SIS draw must perform the row-level dynamic program repeatedly to construct a complete dense matrix, whereas \snake reuses the current state and changes only the loop found by one transition.

\subsection{Directed fixed-degree reciprocity}
We next test the directed-graph extension from Section \ref{subsec: extension}. The implemented \dsnake kernel excludes all label-diagonal cells from the row-zero candidate sets and from the initial-cell distribution, so no transition can create a self-loop. As a classical comparator, we use the directed edge-swap algorithm: two rows and two columns are proposed, and the corresponding checkerboard switch is accepted only when all four cells are admissible and the two proposed directed edges are absent.

The statistic is the number of reciprocal dyads, $\sum_{i<j} A_{ij}A_{ji}$, a standard directed-network measure that is not fixed by the in- and out-degree sequences. We use a $120$-node directed ring in which every node sends edges to its four nearest clockwise and four nearest counterclockwise neighbors. Thus all in- and out-degrees equal $8$, and the observed graph has $480$ reciprocal dyads. This is an intentionally high-reciprocity graph used to test the fixed-degree null distribution.

For the efficiency comparison, both samplers start from the same \dsnake-warmed fixed-degree matrix with $29$ reciprocal dyads. We then run $12$ independent runs per method, each with $30{,}000$ transition attempts, $5{,}000$ burn-in attempts, and spacing $5$, retaining $60{,}000$ post-burn-in reciprocity values per method in total. The two methods give similar fixed-degree null means: $32.01$ for \dsnake, with run-to-run standard error $0.11$, and $31.65$ for directed swap, with run-to-run standard error $0.79$. No retained sample from either method reaches the observed value $480$, giving the plus-one upper-tail estimate $(0+1)/(60000+1)=1.67\times 10^{-5}$ for each method. The observed reciprocal ring is therefore far outside the fixed-degree null distribution.

The efficiency diagnostics are substantially different. Using the initial-positive autocorrelation estimator for the retained reciprocity sequence, \dsnake averages $4965.7$ effective samples per second, with run-to-run standard error $236.4$, whereas directed swap averages $311.7$ effective samples per second, with run-to-run standard error $40.0$. The move-size diagnostic is even more separated: \dsnake flips $12.761$ admissible entries per attempt on average, with standard error $0.016$, while directed swap flips $0.0314$ entries per attempt, with standard error $0.0008$. Figure \ref{fig:dsnake_reciprocity} summarizes the null distribution and the replicated efficiency comparison. These results support \dsnake as a practical constrained sampler for this directed fixed-degree problem, while sharp mixing-time theory for diagonal-constrained margins remains open.

\begin{figure}[htbp!]
\includegraphics[width = \textwidth]{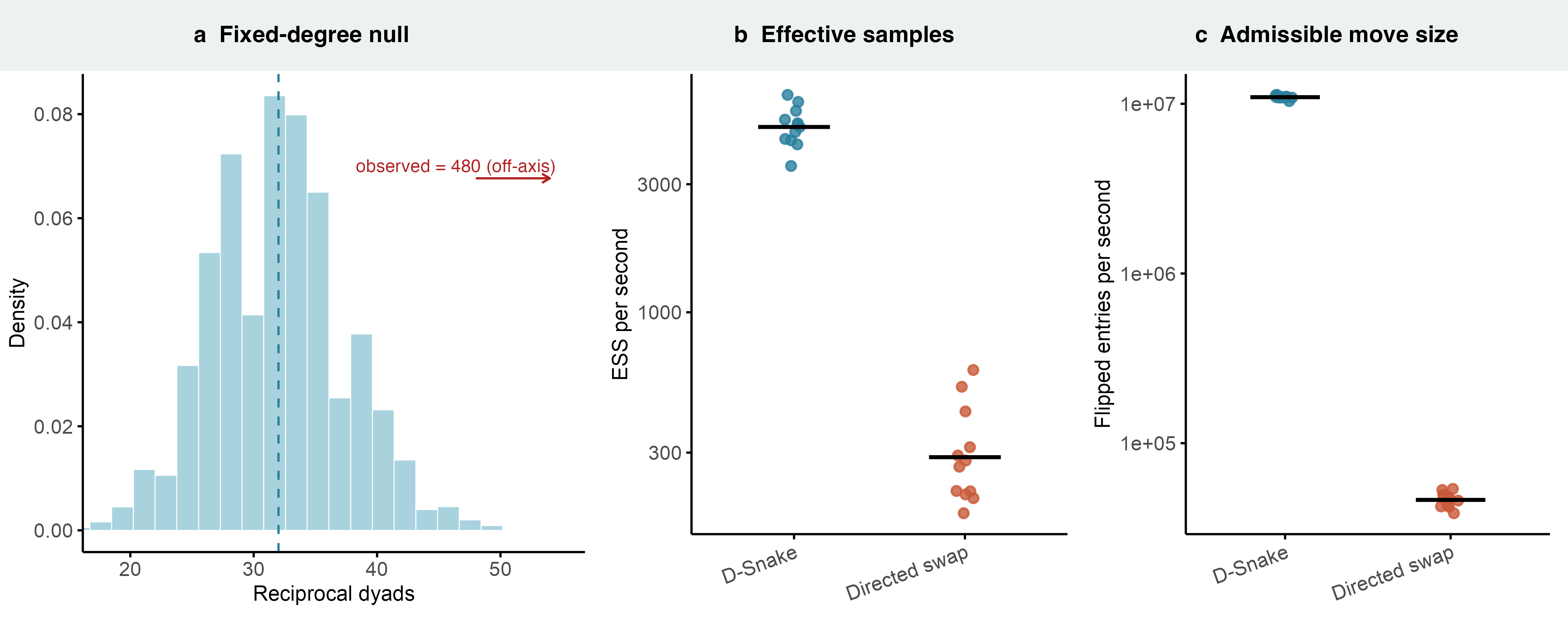}
\caption{Directed fixed-degree reciprocity comparison. Left: fixed-degree null distribution of reciprocal dyads from \dsnake runs; the observed reciprocal ring has $480$ reciprocal dyads and is off-axis. Middle: reciprocity effective samples per second across $12$ runs per method. Right: admissible flipped entries per second across the same runs. Horizontal bars show run means.}
\label{fig:dsnake_reciprocity}
\end{figure}

\section{Conclusion}
This paper develops the \snake algorithm, a rejection-free MCMC sampler for binary matrices with fixed margin sums. The raw chain is irreducible and reversible with respect to the uniform fixed-margin distribution. An important property is that \snake makes provably larger nonlocal moves than local swap-type methods in sparse and balanced regimes, while its work per flipped entry is rate optimal in those regimes and near-optimal up to a polylogarithmic factor under a one-sided half-balanced condition. Combining our Markov-chain comparison with the universal swap-chain spectral-gap theorem of \citet{fu2026spectral} shows that the lazy \snake chain is rapidly mixing for every feasible pair of margins, and in the permutation-matrix case the raw chain has the much sharper total-variation mixing time $\Theta(\sqrt n\log n)$. The numerical experiments support these theoretical conclusions and show favorable performance against the leading MCMC and SIS samplers in this literature.

Looking forward, we mention two computational problems and one theoretical problem. The first computational problem, motivated by applications in differential privacy, is non-uniform sampling of binary matrices; informed proposals in the spirit of \citet{zanella2020informed} may be useful when the target distribution is highly skewed. The second is sampling high-dimensional binary tensors under margin constraints, where swap-based algorithms are known to be very inefficient for $I\times J\times K$ tensors; we hope the loop-finding idea developed here can aid the design of more efficient tensor samplers. On the theoretical side, although Theorem \ref{thm: rapid mixing} gives a universal polynomial bound and Theorem \ref{thm:perm-mixing} gives the sharp order in the permutation-matrix case, sharp \snake mixing bounds for broader fixed-margin families remain open. Closing the gap between the universal comparison estimate and margin-specific behavior is an important direction for future work.

\section*{Acknowledgement}
Guanyang Wang was partially supported by NSF grants DMS-2210849
and CCF-2403007. Peng Zhang was partially supported by NSF CAREER
Award CCF-2238682.

\appendix
\section{Proofs of the theoretical results}\label{app:proofs}
\subsection{Path notation for one step of Snake}\label{app:path-notation}
A \snake trajectory starting from a zero entry in $M$ corresponds exactly to a Snake trajectory starting from a one entry in \((1-M)^\top\). Therefore, the path notation below is given for trajectories whose initial entry is a one. Whenever both the assumptions and the conclusion are invariant under the map $M \mapsto (1-M)^\top$, the other case follows by applying the same argument to $(1-M)^\top$. Otherwise, the case will be treated separately when needed.

Fix a current matrix $M\in \Sigma(\row,\col)$ and let
\[
E=\{(i,j):M(i,j)=1\}.
\]
We write one oriented version of the \snake trajectory in terms of alternating column and row indices. For a column $a\in[n]$, let $f(a)$ be chosen uniformly from $\{i:(i,a)\in E\}$. For a row $b\in[m]$, let $g(b)$ be chosen uniformly from $\{j:(b,j)\notin E\}$. Let $\{f_k\}_{k\ge1}$ and $\{g_k\}_{k\ge1}$ be independent copies of these random maps. Starting from an arbitrary column-valued random variable $a_1$, define
\[
 b_k=f_k(a_k),\qquad a_{k+1}=g_k(b_k),\qquad k\ge1.
\]
Then the entries $(b_k,a_k)$ and $(b_k,a_{k+1})$ are respectively a one and a zero of $M$, and form the same alternating trajectory as Algorithm \ref{alg:Snake}. The estimates below are uniform in the starting distribution of $a_1$. 

Let
\[
R=\inf\{t\ge2:a_t\in\{a_1,\ldots,a_{t-1}\}\},\qquad
R'=\inf\{t\ge2:b_t\in\{b_1,\ldots,b_{t-1}\}\}.
\]
Let $R_0<R$ and $R'_0<R'$ be the unique indices satisfying $a_{R_0}=a_R$ and $b_{R'_0}=b_{R'}$, respectively. The loop found by Algorithm \ref{alg:Snake} has length
\[
L=\begin{cases}
2(R-R_0),& R\le R',\\
2(R'-R'_0),& R>R'.
\end{cases}
\]
The number $C$ of iterations in the while loop satisfies $C\asymp \min(R,R')$, with absolute constants. This is the notation used in the proofs below.

\subsection{Proof of Theorem \ref{thm: correctness}}\label{app:proof-correctness}
\begin{proof}
The termination claim follows from the construction. Before the algorithm stops, no row coordinate can repeat at a vertical closing time and no column coordinate can repeat at a horizontal closing time. Since there are only $m$ rows and $n$ columns, a repeat must occur after at most $m+n+1$ visited entries. When the first repeat occurs, the segment between the earlier and later occurrence is an alternating loop, and flipping that loop leaves every row and column sum unchanged.

We next prove irreducibility. Let $A,B\in\Sigma(\row,\col)$. By Ryser's swap connectivity theorem \citep[Theorem 3.1]{ryser1957combinatorial}, $A$ can be transformed into $B$ by finitely many checkerboard swaps. A checkerboard swap is an alternating loop of length four. For any fixed valid swap, the \snake algorithm can realize that same swap by starting from one corner of the checkerboard and then choosing the other three corners in order. Each choice has positive probability. Hence every swap move has positive \snake probability. Concatenating the swap path gives $\bP_\snake^N(A,B)>0$ for some $N$, so the chain is irreducible.

	It remains to prove reversibility with respect to the uniform law. Since the target law is uniform, it is enough to show that $\bP_\snake(A,B)=\bP_\snake(B,A)$. If $A$ and $B$ do not differ by a single alternating loop, both probabilities are zero. Suppose instead that flipping an alternating loop $\ell$ in $A$ gives $B$. A one-step history from $A$ to $B$ can be written as
	\[
	h=(s_1,\ldots,s_q;p),\qquad 1\le p<q,
	\]
	where $s_t=(i_t,j_t)$ are the visited entries, no earlier segment closes an alternating loop, and the first closing segment $(s_p,\ldots,s_q)$ is exactly the loop that is flipped. Define $h^\leftarrow$ by keeping the pre-closing trajectory $(s_1,\ldots,s_{p-1})$ and replacing the closing segment by the unique reversed cyclic ordering of $\ell$ that can be appended after $s_{p-1}$; if $p=1$, take the reversed segment to start at $s_q$. This reversed segment is well-defined because every vertex of an alternating loop has exactly two loop-neighbors, one sharing its row and one sharing its column, and flipping the loop interchanges the zero/one status needed for traversing the loop in the opposite direction. The prefix entries are not flipped, so the prefix remains a legal \snake trajectory under $B$; at the junction, the first reversed loop entry is the unique loop-neighbor of the last prefix entry with the required shared coordinate and opposite value. Applying the same operation again restores the original cyclic ordering of $\ell$, so the map is an involution. Finally, if $h^\leftarrow$ had an earlier closing segment, reversing that segment by the same construction would give an earlier closing segment in $h$, contradicting the definition of $h$ as a first-closing history.

	For such a paired history, the initial entry has probability $1/(mn)$ in both directions. At each row move, the transition probability is the reciprocal of the number of zeros in the current row; after flipping $\ell$, the reversed move uses the same row, the same row sum, and hence the same number of zeros. At each column move, the transition probability is the reciprocal of the number of ones in the current column, and the same column sum gives the same denominator in the reversed move. The selected entries are the paired corners of the same loop, so the numerators are all one in both directions. Thus the probability weight of $h$ under $A$ equals the probability weight of $h^\leftarrow$ under $B$. Summing over all first-closing histories moving $A$ to $B$ proves symmetry of $\bP_\snake$.

Irreducibility and reversibility imply that $U(\Sigma(\row,\col))$ is the unique stationary distribution.
\end{proof}

\subsection{Proof of Proposition \ref{prop:dsnake-correctness}}\label{app:proof-dsnake}
\begin{proof}
All entries sampled by Algorithm \ref{alg:D-Snake} lie in $(U\times V)\setminus D$, so every realized trajectory avoids the structural diagonal. Since preprocessing removes rows and columns whose admissible entries are completely determined, a current one-entry always has at least one admissible zero in its row and a current zero-entry always has at least one one in its column. The same finite-coordinate argument used for Algorithm \ref{alg:Snake} therefore gives almost-sure termination at a first admissible alternating loop. Flipping such a loop changes one zero to one and one one to zero in every row and column touched by the loop and never touches $D$, so the chain remains in $\Sigma_D(\row,\col;U,V)$.

For reversibility, repeat the first-closing-history reversal from the proof of Theorem \ref{thm: correctness}, restricting all histories to admissible non-diagonal entries. The initial entry has probability $1/(|U||V|-|D|)$ in both directions. Along a row move, the denominator is the number of admissible zeros in the current row; after flipping the loop, the reversed row move sees the same row with the same number of admissible zeros. Along a column move, the denominator is the number of ones in the current column, which is also unchanged by flipping a balanced loop. Thus paired forward and reversed histories have equal probability, and the kernel is symmetric on $\Sigma_D(\row,\col;U,V)$. Symmetry is equivalent to reversibility for the uniform distribution.

It remains only to justify irreducibility under the stated connectivity condition. A directed $2$-switch is an admissible alternating loop of length four in the bipartite matrix representation, and a directed $3$-cycle reorientation is an admissible alternating loop of length six. Algorithm \ref{alg:D-Snake} can realize any fixed such loop with positive probability by starting at one of its entries and following the loop in order. Therefore every edge in the connected directed fixed-degree state graph has positive \dsnake transition probability. Concatenating these positive-probability moves connects any two states in $\Sigma_D(\row,\col;U,V)$, proving irreducibility. Irreducibility and reversibility on the finite state space give the uniform law as the unique stationary distribution.
\end{proof}

\subsection{Proof of Theorem \ref{thm: flips-sparse}}\label{app:proof-flips-sparse}

In the notation of Appendix~\ref{app:path-notation}, we first estimate the quantity $R - R_0$.

\begin{lemma}\label{lem: r}
    $$\bP\left(R-R_0\leq\frac{1}{32}\min\left\{\frac{n}{\rmax},\sqrt n\right\}\right)\leq\frac14.$$
\end{lemma}
\begin{proof}
    If $\rmax > \frac{n}{2}$, the lemma holds. From now on, we assume that $\rmax \le \frac{n}{2}$.
    
    Conditioned on $a_1,\ldots,a_k,b_1,\ldots,b_k$, the next column $a_{k+1}$ is chosen uniformly from the zeros in row $b_k$. Thus, we have
    \begin{align*}
        &\bP(R=k+1\mid R>k)\\
        =&\bE(\bP(R=k+1\mid a_1,\ldots,a_k,b_1,\ldots,b_k,R>k)\mid R>k)\\
        =&\bE\left(\frac{
    |\{1\leq l\leq k:M(b_k,a_l)=0\}}{n-r_{b_k}}\mid R>k)\right)\\
\ge& \bE\left(\frac{k-r_{b_k}}{n-r_{b_k}}
\mid R>k
\right)\\
\ge&\frac{k-\rmax}{n-\rmax}.
    \end{align*}

Therefore, for $s:= \frac{1}{32} \min\left\{\frac{n}{\rmax},\sqrt n\right\},$
we have \begin{align*}
    \bP(R-R_0\leq s)=&\sum_{k=1}^{n}\bP(R=k+1,\ R-R_0\leq s)\\
\leq& \sum_{k=1}^{n}\bP\left(a_{k+1}\in\{a_k,\ldots,a_{\max\left\{1,k-\lfloor s\rfloor+1\right\}}\}, R>k \right)\\
=&
\sum_{k=1}^{n} \bP\left(a_{k+1}\in\{a_k,\ldots,a_{\max\left\{1,k-\lfloor s\rfloor+1\right\}}\} \mid R>k \right)
\bP(R>k)\\
\leq& \frac{ s}{n-\rmax} \sum_{k=1}^{n}\bP(R>k)\\
\leq& \frac{ s}{n-\rmax} \left(\rmax + 8s+\sum_{k=\rmax +\lceil 8s\rceil}^{n}\bP(R>k)\right)\\
=& \frac{ s}{n-\rmax} \left(\rmax + 8s+\sum_{k=\rmax +\lceil 8s\rceil}^{n}\frac{\bP(R=k+1)}{\bP(R=k+1\mid R>k)}\right)\\
\le& \frac{ s}{n-\rmax} \left(\rmax + 8s+\frac{n-\rmax}{8s}\sum_{k=\rmax +\lceil 8s\rceil}^{n}\bP(R=k+1)\right).\\
\leq& \frac{ s(\rmax+8s)}{n-\rmax} +\frac{1}{8}\\
\leq& \frac{1}{4}.
\end{align*}

\end{proof}

Next, we estimate the quantity $R'-R_0'$.

\begin{lemma}\label{lem: r'}
\[
    \bP\left( R'-R'_0\leq\frac{1}{32}\min\left\{\frac{n}{\rmax}, \sqrt{\frac{n \cmin}{\rmax}}\right\}\right)\leq\frac14.
\]
\end{lemma}
\begin{proof}
    If $\rmax > \frac{n}{2}$, the lemma holds. From now on, we assume that $\rmax \le \frac{n}{2}$.

    For each $1\le i\le m$, define $$p_i:= \sum_{M(i,j)=1}\frac{1}{nc_j}.$$
    Then for each $1\le i\le m$, $$p_i\le \frac{r_i}{n\cmin }\le \frac{\rmax}{n\cmin}.$$

    Conditioned on $b_k$, the next row $b_{k+1}$ has distribution
\begin{align*}
\bP(b_{k+1}=i\mid b_k)
&=\sum_{\substack{M(b_k,j)=0\\M(i,j)=1}}\frac{1}{(n-r_{b_k})c_j}\\
&\leq 2\sum_{M(i,j)=1}\frac{1}{nc_j}\\
&=2p_i.
\end{align*}

For a positive real number $s$, let $S$ be the largest nonnegative integer such that 
$$\sum_{k=1}^{S}p_{b_k}\le\frac{\rmax}{n}\left(1+\frac{16s}{\cmin}\right).$$
Since the event $\{j\leq S\}$ is determined by $b_1,\ldots,b_j$, we have 
\begin{align*}
    &\bP(R'-R'_0\le s, R'_0\le S)\\
    \leq&\sum_{i=1}^{m}\sum_{j=1}^{\infty}\sum_{k=1}^{\lfloor s\rfloor}
    \bP( b_j=i, b_{j+k}=i,j\leq S)\\
=&\sum_{i=1}^{m}\sum_{j=1}^{\infty}\sum_{k=1}^{\lfloor s\rfloor}
\bP(b_j=i,j\leq S)
\bP(b_{j+k}=i\mid b_j=i,j\le S)\\
\le&2s\sum_{i=1}^{m}\sum_{j=1}^{\infty}p_i\bP(b_j=i, j\leq S)\\
=&2s\bE\left(\sum_{j=1}^{S}p_{b_j}\right)\\
\le&
\frac{2s\rmax}{n}
\left(
    1+\frac{16s}{\cmin}
\right).
\end{align*}

Conditioned on $R'>k$, the rows $b_1,\ldots,b_k$ are distinct. So we have
\begin{align*}
&\bP(R'=k+1\mid b_1,\ldots,b_k,\ R'>k>S)\\
=&\sum_{i=1}^{k}\sum_{\substack{M(b_k,j)=0\\M(b_i,j)=1}}
\frac{1}{(n-r_{b_k})c_j}\\
\ge&
\sum_{i=1}^{k}\sum_{\substack{M(b_k,j)=0\\M(b_i,j)=1}}\frac{1}{nc_j}\\
=&
\sum_{i=1}^{k}p_{b_i}-\sum_{M(b_k,j)=1}
\frac{|\{1\leq i\leq k:M(b_i,j)=1\}|}{nc_j}\\
\ge&
\sum_{i=1}^{k}p_{b_i}-\frac{r_{b_k}}{n}\\
\ge&
\sum_{i=1}^{k}p_{b_i}-\frac{\rmax}{n}\\
>&
\frac{16s\rmax}{n\cmin}.
\end{align*}

On the other hand, conditioned on the same history, we have
\begin{align*}
&\bP(R'=k+1, R'-R'_0\le s\mid b_1,\ldots,b_k,\ R'>k>S)\\
=&\bP\left(b_{k+1}\in\{b_k,\ldots,b_{\max\left\{1,k-\lfloor s\rfloor+1\right\}}\}\mid b_1,\ldots,b_k, R'>k>S \right)\\
\le & 2\sum_{i=\max\left\{1,k-\lfloor s\rfloor+1\right\}}^{k}p_{b_i}\\
\le&
\frac{2s\rmax}{n\cmin}.
\end{align*}

Then we have 
\begin{align*}
    &\bP(R'-R_0'\le s,R_0'>S)\\
    =&\sum_{k=1}^m\bP(R'=k+1, R'-R_0'\le s,R_0'>S)\\
\leq& \sum_{k=1}^{m}\bP\left(R'=k+1 , R'-R_0'\le s, R'>k>S\right)\\
=&
\sum_{k=1}^{m} \bP\left(R'=k+1 , R'-R_0'\le s \mid R'>k>S\right)\bP(R'>k>S)\\
\le& \frac{2s\rmax}{n\cmin}\sum_{k=1}^{m}\bP(R'>k>S)\\
=& \frac{2s\rmax}{n\cmin}\sum_{k=1}^{m}\frac{\bP(R'=k+1,R'>k>S)}{\bP(R'=k+1\mid R'>k>S)}\\
\le & \frac{1}{8}\sum_{k=1}^{m}\bP(R'=k+1)\\
\le&\frac{1}{8}.
\end{align*}

Therefore, for $s:=\frac{1}{32}\min\left\{\frac{n}{\rmax}, \sqrt{\frac{n \cmin}{\rmax}}\right\}$, we have 
\begin{align*}
    &\bP(R'-R'_0\leq s)\\
=&\bP(R'-R_0'\le s,R_0'\le S)+\bP(R'-R_0'\le s, R_0'>S)\\
\le&
\frac{2s\rmax}{n}
\left(1+\frac{16s}{\cmin}\right)
+\frac18\\
\le &\frac14.
\end{align*}
\end{proof}

Now we are ready to prove Theorem~\ref{thm: flips-sparse}.

\begin{proof}[Proof of Theorem~\ref{thm: flips-sparse}]

We first prove $$\bE(L)\ge \frac{1}{64}\sqrt{\frac{n}{\rmax}}.$$

If the initial entry is a one, then $L\ge2\min\{R-R_0,R'-R_0'\}$ in the notation of Appendix~\ref{app:path-notation}. By Lemma~\ref{lem: r} and Lemma~\ref{lem: r'}, we have 
\begin{align*}
    \bE(L)\ge &\frac{1}{16}\sqrt{\frac{n}{\rmax}}\bP\left(L>\frac{1}{16}\sqrt{\frac{n}{\rmax}}\right)\\
    \ge &\frac{1}{16}\sqrt{\frac{n}{\rmax}}\bP\left(2\min\{R-R_0,R'-R_0'\}>\frac{1}{16}\sqrt{\frac{n}{\rmax}}\right)\\
    \ge &\frac{1}{16}\sqrt{\frac{n}{\rmax}}\left(1-\bP\left(R-R_0\le\frac{1}{32}\sqrt{\frac{n}{\rmax}}\right)-\bP\left(R'-R'_0\le\frac{1}{32}\sqrt{\frac{n}{\rmax}}\right)\right)\\
    \ge &\frac{1}{32}\sqrt{\frac{n}{\rmax}}.
\end{align*}

If the initial entry is zero, we denote it as $(b_0,a_1)$. Define the sequences $\{a_k\}_{k\in\mathbb{N}}$ and $\{b_k\}_{k\in\mathbb{N}}$ and $R,R',R_0,R_0'$ as in the notation of Appendix~\ref{app:path-notation}. If the loop length $L$ is not equal to $2(R-R_0)$ or $2(R'-R_0')$, then $b_{\frac{L}{2}}=b_0$. 

By the proof of Lemma~\ref{lem: r'}, we have
\begin{align*}
&\bP\left(b_k=b_0\text{ for some }1\leq k\le  \left\lfloor \frac{1}{32}\sqrt{\frac{n}{\rmax}}\right\rfloor\right)\\
\le&
2p_{b_0}\left\lfloor \frac{1}{32}\sqrt{\frac{n}{\rmax}}\right\rfloor\\
\le&
\frac{2\rmax}{n\cmin}\left\lfloor \frac{1}{32}\sqrt{\frac{n}{\rmax}}\right\rfloor\\
\le&\frac1{16}.
\end{align*}

Thus, by Lemma~\ref{lem: r} and Lemma~\ref{lem: r'}, we have 
\begin{align*}
    \bE(L)\ge &\frac{1}{16}\sqrt{\frac{n}{\rmax}}\bP\left(L>\frac{1}{16}\sqrt{\frac{n}{\rmax}}\right)\\
    \ge &\frac{1}{16}\sqrt{\frac{n}{\rmax}}\left(\bP\left(2\min\{R-R_0,R'-R_0'\}>\frac{1}{16}\sqrt{\frac{n}{\rmax}}\right)-\frac{1}{16}\right)\\
    \ge &\frac{1}{16}\sqrt{\frac{n}{\rmax}}\left(\frac{15}{16}-\bP\left(R-R_0\le\frac{1}{32}\sqrt{\frac{n}{\rmax}}\right)-\bP\left(R'-R'_0\le\frac{1}{32}\sqrt{\frac{n}{\rmax}}\right)\right)\\
    \ge &\frac{1}{64}\sqrt{\frac{n}{\rmax}}.
\end{align*}

By applying the previous bound on $\overline M=(1-M)^\top$, we have $$\bE(L)\ge \frac{1}{64}\sqrt{\frac{m}{m-\cmin}}.$$
\end{proof}

\subsection{Sparse case of Theorem \ref{thm: rate optimality}}\label{app:proof-sparse-rate}
\begin{proof}
Assume $\rmax$ and $\cmax$ are bounded by a fixed constant $K$. Since every row and every column has at least one one after removing degenerate rows and columns,
\[
 m\le \sum_{i=1}^m r_i=\sum_{j=1}^n c_j\le Kn,
 \qquad
 n\le \sum_{j=1}^n c_j=\sum_{i=1}^m r_i\le Km.
\]
Thus $m\asymp_K n$.

We now bound the cost. In the notation of Appendix \ref{app:path-notation}, suppose no column repeat has occurred among $a_1,\ldots,a_k$. Conditional on the history up to $b_k$, the next column $a_{k+1}=g_k(b_k)$ is uniform over the zeros in row $b_k$. This row has at most $K$ ones, so at least $(k-K)_+$ of the previously seen columns are available zeros in that row. Therefore
\[
\bP(R>k+1\mid R>k)\le 1-\frac{(k-K)_+}{n}.
\]
Iterating this bound gives, for $t>2K$,
\[
\bP(R>t)\le \exp\left(-\sum_{k=K+1}^{t-1}\frac{k-K}{n}\right)
\le \exp\left(-\frac{(t-K-1)^2}{4n}\right).
\]
Consequently,
\[
\bE(R)=\sum_{t\ge0}\bP(R>t)=O_K(\sqrt n).
\]
If the initial entry is zero, then after the first iteration, the trajectory takes the same form used above. Thus, the same estimate holds with at most one additional inner-loop iteration. Since the number $C$ of inner-loop iterations is bounded by an absolute constant times $\min(R,R')$, we have $\bE(C)=O_K(\sqrt n)$. Theorem \ref{thm: flips-sparse} gives $\bE(L)=\Omega_K(\sqrt n)$ because $m\asymp_K n$. Also $L\le C$ up to an absolute constant by construction. Hence $\bE(C)/\bE(L)=\Theta_K(1)$, which proves the sparse case of Theorem \ref{thm: rate optimality}.
\end{proof}

\subsection{Proof of Theorem \ref{thm: upper bound}}\label{app:proof-upper-bound}
\begin{proof}
	Let $t_1$ and $t_2$ be two positive integers to be determined later. Let $x_1$ be an arbitrary deterministic integer in $\{1,\ldots, n\}$. Let the sequences $\{x_k\}$ and $\{y_k\}$ be generated by $y_k=f_{t_1 t_2+k}(x_k)$ and $x_{k+1}= g_{t_1 t_2+k}(y_k)$. Then the random variables $x_k$ and $y_k$ ($1\le k\le t_1 t_2$) are independent with the random variables $a_k$ and $b_k$ ($1\le k\le t_1 t_2$). For integers $0\le k,l\le 2t_1 t_2$, let $A_{k,l}$ denote the event that
	\begin{enumerate}[label=(\roman*)]
	    \item there are no repeated values in $a_1, \ldots, a_{\left\lceil \frac{k}{2}\right\rceil}, x_1, \ldots, x_{\left\lceil \frac{l}{2}\right\rceil}$, and
	    \item there are no repeated values in $b_1, \ldots, b_{\left\lfloor \frac{k}{2}\right\rfloor}, y_1, \ldots, y_{\left\lfloor \frac{l}{2}\right\rfloor}$.
	\end{enumerate}
	
	Let $s$ be an integer with $1\le s\le t_1$. When $l$ is odd, we have
	\begin{align*}
	    &\mathbb{P}\left(A_{2st_2,l-1}\setminus A_{2st_2,l} \right) \\
	    \ge &\mathbb{P}\left(A_{2st_2,l-1} \mbox{ and } x_{\frac{l+1}{2}} \in \left\{a_1,\ldots, a_{st_2} \right\} \right)\\
	    =&\mathbb{E}\left(\mathbb{P}\left(A_{2st_2,l-1} \mbox{ and } x_{\frac{l+1}{2}} \in \left\{a_1,\ldots, a_{st_2} \right\}\;\middle| \; a_1,\ldots, a_{st_2},y_1,\ldots,y_{\frac{l-1}{2}}\right)\right)\\
	    =&\mathbb{E}\left(\mathbb{E}\left(\frac{\left|\left\{1\le i\le st_2:\left(y_{\frac{l-1}{2}},a_i\right)\not \in E\right\}\right|}{n-r_{y_{\frac{l-1}{2}}}}\;\mathbf{1}_{A_{2st_2,l-1}} \;\middle| \; a_1,\ldots, a_{st_2},y_1,\ldots,y_{\frac{l-1}{2}}\right)\right)\\
	    =&\mathbb{E}\left(\frac{\left|\left\{1\le i\le st_2:\left(y_{\frac{l-1}{2}},a_i\right)\not \in E\right\}\right|}{n-r_{y_{\frac{l-1}{2}}}}\;\mathbf{1}_{A_{2st_2,l-1}}\right)\\
	    \ge &\frac{1}{n}\;\mathbb{E}\left(\left|\left\{1\le i\le st_2:\left(y_{\frac{l-1}{2}},a_i\right)\not \in E\right\}\right|\;\mathbf{1}_{A_{2st_2,l-1}}\right)\\
	    = &\frac{1}{n}\sum_{i=1}^{st_2}\mathbb{P}\left(A_{2st_2,l-1}\mbox{ and }\left(y_{\frac{l-1}{2}},a_i\right)\not \in E\right).
	\end{align*}
	And when $k$ is even, we have
	\begin{align*}
	    &\mathbb{P}\left(A_{k-1,2st_2}\setminus A_{k,2st_2} \right)\\
	    \ge&\mathbb{P}\left(A_{k-1, 2st_2} \mbox{ and } b_{\frac{k}{2}} \in\left\{y_1,\ldots , y_{st_2}\right\}\right)\\
	    =&\mathbb{E}\left(\mathbb{P}\left(A_{k-1, 2st_2} \mbox{ and } b_{\frac{k}{2}} \in\left\{y_1,\ldots , y_{st_2}\right\}\;\middle|\; a_1,\ldots, a_{\frac{k}{2}}, y_1,\ldots, y_{st_2}\right)\right)\\
	    =&\mathbb{E}\left(\mathbb{E}\left(\frac{\left|\left\{1\le j\le st_2: \left(y_j,a_{\frac{k}{2}}\right)\in E\right\}\right|}{c_{a_{\frac{k}{2}}}}\;\mathbf{1}_{A_{k-1,2st_2}}\;\middle|\; a_1,\ldots, a_{\frac{k}{2}}, y_1,\ldots, y_{st_2} \right)\right)\\
	    =&\mathbb{E}\left(\frac{\left|\left\{1\le j\le st_2: \left(y_j,a_{\frac{k}{2}}\right)\in E\right\}\right|}{c_{a_{\frac{k}{2}}}}\;\mathbf{1}_{A_{k-1,2st_2}}\right)\\
	    \ge& \frac{1}{m}\;\mathbb{E}\left(\left|\left\{1\le j\le st_2: \left(y_j,a_{\frac{k}{2}}\right)\in E\right\}\right|\; \mathbf{1}_{A_{k-1,2st_2}}\right)\\
	    =&\frac{1}{m}\sum_{j=1}^{st_2} \mathbb{P}\left(A_{k-1,2st_2}\mbox{ and }\left(y_j, a_{\frac{k}{2}}\right) \in E\right).
	\end{align*}
	Then we have
	\begin{align*}
	    &(m+n) \mathbb{P}\left(A_{2(s-1)t_2,2 (s-1)t_2}\right)\\
	    \ge&n\; \mathbb{P}\left(A_{2 st_2,2(s-1) t_2+2}\setminus A_{2st_2,2st_2-1}\right)+m\;\mathbb{P}\left(A_{2(s-1)t_2+1,2st_2}\setminus A_{2st_2,2st_2}\right)\\
	    \ge&n\sum_{j=(s-1)t_2+1}^{st_2-1} \mathbb{P}\left(A_{2st_2,2j}\setminus A_{2st_2,2j+1}\right)+m\sum_{i=(s-1)t_2+1}^{st_2}\mathbb{P}\left(A_{2i-1,2st_2}\setminus A_{2i,2st_2}\right)\\
	    \ge &\sum_{j=(s-1)t_2+1}^{st_2-1}\sum_{i=1}^{st_2}\mathbb{P}\left(A_{2st_2,2j}\mbox{ and }\left(y_{j},a_i\right)\not \in E\right)\\&+ \sum_{i=(s-1)t_2+1}^{st_2}\sum_{j=1}^{st_2} \mathbb{P}\left(A_{2i-1,2st_2}\mbox{ and }\left(y_j, a_{i}\right) \in E\right)\\
	    \ge& t_2 (t_2-1)  \mathbb{P}\left(A_{2st_2,2st_2}\right).
	\end{align*}
	Take $t_2= \left\lfloor 3\sqrt{m+n}\right\rfloor$, then we have $$\mathbb{P}\left(A_{2t_1t_2,2t_1t_2}\right)\le \frac{1}{2^{t_1}}.$$
	Take $t_1= \left\lfloor 5\log\left( m+n\right)\right\rfloor$, then we have $$\mathbb{P}\left(A_{2t_1t_2,2t_1t_2}\right)\le \frac{1}{(m+n)^2}.$$
	
	Note that $x_k=a_{t_1t_2+k}$ and $y_k=b_{t_1t_2+k}$ when $x_1=a_{t_1t_2+1}$. Using the union bound we conclude that
	\begin{align*}
	&\mathbb{E}\left(\min\left(R,R'\right)\right)\\
	    \le&2t_1t_2+(m+n)\mathbb{P}\left(\min\left(R,R'\right)> 2t_1t_2\right)\\
	    \le &2t_1t_2+(m+n)\mathbb{P}\left(A_{2t_1t_2,2t_1t_2}\mbox{ for some } x_1\in\left\{1,\ldots, n\right\}\right)\\
	    \le &2t_1t_2+(m+n)\sum_{i=1}^n\mathbb{P}\left(A_{2t_1t_2,2t_1t_2}\mbox{ for } x_1= i\right)\\
	    \le &2t_1t_2+\frac{n}{m+n}\\
	    \le &32\sqrt{m+n}\log(m+n)
	\end{align*}
	Since $C\asymp \min\left(R,R'\right)$, we have $\mathbb{E}(C)=O\left(\sqrt{m+n}\log(m+n)\right)$, and Theorem \ref{thm: upper bound} holds.
\end{proof}
	
	\subsection{Proof of Theorem \ref{thm: flips-balance} and the balanced case of Theorem \ref{thm: rate optimality}}\label{app:proof-balanced}
\begin{proof}
	In the proof of Theorem \ref{thm: upper bound}, we no longer restrict $t_1$ to be $\left\lfloor 5\log\left( m+n\right)\right\rfloor$. Since $r_i\le \Delta n$ for all $1\le i\le m$, we have
	\begin{align*}
	    &\mathbb{P}\left(\min\left(R,R'\right)> 2t_1t_2\right)\\
	    =&\sum_{i=1}^n\mathbb{P}\left(\min\left(R,R'\right)> 2t_1t_2\mbox{ and } a_{t_1t_2+1}=i\right)\\
	    =&\sum_{i=1}^n \mathbb{P}\left( \left(A_{2t_1t_2, 2t_1t_2}\mbox{ for } x_1=i\right)\mbox{ and } a_{t_1t_2+1}=i\right)\\
	    =&\sum_{i=1}^n\mathbb{P}\left( A_{2t_1t_2, 2t_1t_2}\mbox{ for } x_1=i\right)\mathbb{P}\left(a_{t_1t_2+1}=i \;\middle|\;A_{2t_1t_2, 2t_1t_2}\mbox{ for } x_1=i\right)\\
	    \le & \sum_{i=1}^n\mathbb{P}\left( A_{2t_1t_2, 2t_1t_2}\mbox{ for } x_1=i\right)\mathbb{E}\left(\frac{1}{n-r_{b_{t_1t_2}}} \;\middle|\;A_{2t_1t_2, 2t_1t_2}\mbox{ for } x_1=i\right)\\
	    \le & \frac{1}{2^{t_1}(1-\Delta)}.
	\end{align*}
	Thus we have
	\begin{align*}
	    &\mathbb{E}\left(\min\left(R,R'\right)\right)\\
	    \le&2t_2\; \mathbb{E}\left(\left\lceil\frac{\min\left(R,R'\right)}{2t_2}\right\rceil\right)\\
	    =& 2 t_2 \sum_{t_1=0}^\infty \mathbb{P}\left(\left\lceil\frac{\min\left(R,R'\right)}{2t_2}\right\rceil > t_1 \right)\\
	    = & 2 t_2 \left(1+\sum_{t_1=1}^\infty \mathbb{P}\left(\min\left(R,R'\right) >2t_1t_2\right)\right)\\
	    \le&\frac{4t_2}{1-\Delta}\\
	    \le&\frac{12\sqrt{m+n}}{1-\Delta}
	\end{align*}
	Because $m\asymp n$ and $C\asymp \min\left(R,R'\right)$, we have $\mathbb{E}\left(\min\left(R,R'\right)\right)=O(\sqrt{n})$ and $\mathbb{E}(C)=O(\sqrt{n})$.
	
	By the union bound, for any positive integers $s_1$ and $s_2$, we have
	\begin{align*}
	   &\mathbb{P}\left(L\le s_2\mbox{ and } \min\left(R,R'\right)\le s_1\right)\\
	   \le &\mathbb{P}\left(R-R_0\le s_2 \mbox{ and } R\le s_1\right)+\mathbb{P}\left(R'-R_0'\le s_2 \mbox{ and } R'\le s_1\right)\\
	    \le &\sum_{i=1}^{s_1}\sum_{j=1}^{s_2}\mathbb{P}\left(a_i=a_{i+j}\right)+\sum_{i=1}^{s_1}\sum_{j=1}^{s_2}\mathbb{P}\left(b_i=b_{i+j}\right)\\
	    \le &\sum_{i=1}^{s_1}\sum_{j=1}^{s_2}\mathbb{E}\left(\frac{1}{n-r_{b_{i+j-1}}}\right)+\sum_{i=1}^{s_1}\sum_{j=1}^{s_2}\mathbb{E}\left(\frac{1}{c_{a_{i+j}}}\right)\\
	    \le &{s_1 s_2} \left(\frac{1}{(1-\Delta)n}+\frac{1}{\delta m}\right).
	\end{align*}
	By Markov's inequality, we have
	\begin{align*}
	    &\mathbb{E}\left(L\right)\\
	    \ge &s_2\; \mathbb{P}\left(L>s_2\right)\\
	    \ge &s_2 \left(1-\mathbb{P}\left(\min\left(R,R'\right)> s_1\right)-\mathbb{P}\left(L\le s_2\mbox{ and } \min\left(R,R'\right)\le s_1\right)\right)\\
	    \ge& s_2\left(1-\frac{\mathbb{E}\left(\min\left(R,R'\right)\right)}{s_1}-{s_1 s_2} \left(\frac{1}{(1-\Delta)n}+\frac{1}{\delta m}\right) \right).
	\end{align*}
	Take $s_1=\left\lceil\frac{36 \sqrt{m+n}}{1-\Delta}\right\rceil$ and $s_2= \left\lfloor \frac{1 }{3s_1}\left(\frac{1}{(1-\Delta)n}+\frac{1}{\delta m}\right)^{-1}\right\rfloor$, then the two error terms are at most $\frac{1}{3}$, and hence $\mathbb{E}(L) \ge \frac{s_2}{3}$. Since $m \asymp n$, we have $s_2 = \Omega(\sqrt{n})$. Therefore, $\mathbb{E}(L) = \Omega(\sqrt{n})$.
	
	Therefore,
	\[
	\mathbb{E}(C)\asymp \mathbb{E}(L)\asymp \sqrt{n},
	\qquad
	\frac{\mathbb{E}(C)}{\mathbb{E}(L)}=\Theta(1).
	\]
	The first relation proves Theorem \ref{thm: flips-balance}, and the second proves the balanced case of Theorem \ref{thm: rate optimality}.
\end{proof}
	
\subsection{Proof of Theorem \ref{thm:half-balanced-efficiency}}\label{app:proof-half-balanced}
\begin{proof}
Write $N=m+n$ and use the notation of Appendix \ref{app:path-notation}. We first prove the required lower bound for the column closing gap $R-R_0$; the row closing gap is handled at the end by applying the same argument to $(1-M)^\top$.

Let
\[
    \eta=\min\{2-2\delta,2\Delta\}>1,\qquad
    t_1=\left\lceil \frac{20\log N}{\log^2\eta}\right\rceil .
\]
If $\rmax\le \delta n$, then for any column-valued random variables $x$ and $y$,
\begin{align*}
    &\frac12\sum_{i=1}^n
    \left|\bP(g(f(x))=i)-\bP(g(f(y))=i)\right|  \\
    &\qquad\le
    \frac12\sum_{i=1}^n
    \max\{\bP(g(f(x))=i),\bP(g(f(y))=i)\}
    \le \frac{1}{2-2\delta}.
\end{align*}
If $\cmin\ge \Delta m$, the same calculation through the intermediate row gives
\[
    \frac12\sum_{i=1}^n
    \left|\bP(g(f(x))=i)-\bP(g(f(y))=i)\right|
    \le \frac{1}{2\Delta}.
\]
Thus the one-step column kernel has Dobrushin contraction coefficient at most $\eta^{-1}$. By the standard contraction argument for total variation distance \citep[Lemma~4.11]{LP}, if
\[
    q_i:=\bP(a_{t_1+1}=i),\qquad
    \varepsilon:=\frac{1}{100t_1^2N^2},
\]
then, after increasing the numerical constant in the definition of $t_1$ if necessary,
\begin{equation}\label{eq:half-mixing}
    \frac12\sum_{i=1}^n
    \left|\bP(a_{k+t_1}=i\mid a_k)-q_i\right|
    \le \varepsilon
    \qquad\text{for every }k\ge1 .
\end{equation}

We next use a block coupling to compare the chain at times separated by $t_1$ with independent $q$-samples. At the block times $t_1,2t_1,\ldots$, construct random variables
\[
    \tilde a_{t_1},\tilde a_{2t_1},\ldots
\]
such that they are mutually independent with common distribution $q$ and, conditionally on the past up to time $s t_1$,
\begin{equation}\label{eq:block-coupling}
    \bP\left(a_{(s+1)t_1}\ne \tilde a_{s t_1}\mid a_1,\ldots,a_{s t_1}\right)
    \le \varepsilon .
\end{equation}
This follows from maximal coupling and \eqref{eq:half-mixing}; the marginal law of each new $\tilde a_{s t_1}$ is $q$, independent of the past, so the block variables are independent while staying coupled to the corresponding chain endpoints.

Let $T_P$ be Poisson with mean $t_2$, independent of the block variables, and set
\[
    t_2=\left\lfloor
    \frac{10\log N}{\sqrt{\sum_iq_i^2}}
    \right\rfloor .
\]
The standard Poissonized birthday formula \citep{Camarri} gives
\[
    \bP(\tilde a_{t_1},\ldots,\tilde a_{t_1T_P}\text{ all distinct})
    =e^{-t_2}\prod_i(1+q_it_2)
    \le N^{-3}.
\]
Also $\bP(T_P\ge3t_2)\le e^{-t_2}\le N^{-8}$. The fixed-length distinct event is contained in the Poissonized distinct event together with $\{T_P\ge3t_2\}$. Therefore
\begin{align}
    \bP(R>3t_1t_2)
    &\le
    \bP(\tilde a_{t_1},\tilde a_{2t_1},\ldots,
       \tilde a_{(3t_2-1)t_1}\text{ all distinct})
       +\sum_{s=1}^{3t_2-1}\bP(a_{(s+1)t_1}\ne\tilde a_{s t_1}) \notag\\
    &\le N^{-3}+N^{-8}+3t_2\varepsilon
    \le \frac{2}{7N}. \label{eq:first-R-tail}
\end{align}

We next record a second upper bound on $R$. If $\cmin\ge\Delta m$, define
\[
    p_k=\frac1m\sum_{i=1}^m
    \frac{\mathbf 1\{(i,k)\notin E\}}{n-r_i},
    \qquad 1\le k\le n .
\]
Then $\sum_kp_k=1$, and for every column-valued random variable $x$,
\[
    \bP(g(f(x))=k)\le \Delta^{-1}p_k\le 2p_k.
\]
Let $k_0$ maximize $p_k$ and let
\[
    S=\left\{k:\frac1m\sum_{i=1}^m
    \frac{\mathbf 1\{(i,k)\notin E,\ (i,k_0)\notin E\}}{n-r_i}
    \ge \frac{p_{k_0}}4\right\}.
\]
A direct counting estimate gives $|S|\le 2/p_{k_0}$, and if $k\notin S$, then
\[
    \bP(g(f(k))=k_0)\ge \frac{p_{k_0}}2 .
\]
With
\[
    t_3=\left\lceil\frac{20\log N}{(\max_kp_k)^2}\right\rceil ,
\]
a Chernoff bound for variables with conditional success probability at least $p_{k_0}/2$ yields
\begin{equation}\label{eq:second-R-tail-dense}
    \bP(R>2t_3)\le N^{-2}.
\end{equation}
If instead $\rmax\le\delta n$, set $p_k=1/n$ for all $k$. Then $\bP(g(f(x))=k)\le 2p_k$, and with the same definition of $t_3$ we have $R\le n+1\le2t_3$ for all large $N$; finitely many small values of $N$ are absorbed into the constants. Thus \eqref{eq:second-R-tail-dense} holds in either half-balanced case.

Combining \eqref{eq:first-R-tail} and \eqref{eq:second-R-tail-dense}, and using $R\le n+1$, gives
\begin{equation}\label{eq:ER-upper-half}
    \bE (R)
    \le
    \min(3t_1t_2,2t_3)
    +(n+1)(\bP(R>3t_1t_2)+\bP(R>2t_3))
    \le 4\min(t_1t_2,t_3).
\end{equation}

It remains to rule out a very recent first repeat. Let
\[
    d=\left\lfloor\frac{t_2}{30t_1^3}\right\rfloor .
\]
The parameters just defined satisfy
\[
    \frac{t_1}{t_2}\ge 3\sqrt{\sum_iq_i^2},
    \qquad
    \sqrt{\frac{t_1}{t_3}}\ge \max_i p_i\ge \sum_i p_i^2 .
\]
Moreover, by the domination $\bP(g(f(x))=i)\le2p_i$ and \eqref{eq:half-mixing},
\[
    \sum_{j=1}^{3t_1t_2}\bP(a_j=i)
    \le \bP(a_1=i)+2t_1p_i+3t_1t_2(q_i+\varepsilon),
\]
and, uniformly in $j$,
\[
    \sum_{k=1}^{d}\bP(a_{j+k}=i\mid a_j=i)
    \le 2t_1p_i+d(q_i+\varepsilon).
\]
The term $\bP(a_1=i)$ accounts for the arbitrary starting law of $a_1$. Therefore,
\begin{align}
    \bP(R-R_0\le d)
    &\le \bP(R>3t_1t_2)
      +\sum_i\sum_{j=1}^{3t_1t_2}\sum_{k=1}^d
        \bP(a_j=i)\bP(a_{j+k}=i\mid a_j=i) \notag\\
    &\le
      \frac17+
      \sum_i(2t_1p_i+3t_1t_2(q_i+\varepsilon))
            (2t_1p_i+d(q_i+\varepsilon)) \notag\\
    &\hspace{0.35in}
      +\max_i\{2t_1p_i+d(q_i+\varepsilon)\}. \label{eq:recent-collision-expanded}
\end{align}
Expanding the sum in \eqref{eq:recent-collision-expanded} and using the two preceding parameter bounds gives
\[
    \frac17+
      \sum_i(2t_1p_i+3t_1t_2(q_i+\varepsilon))
            (2t_1p_i+d(q_i+\varepsilon))
    \le
    \frac{5}{11}+4t_1^{13/4}t_3^{-1/4}.
\]
For the final maximum term, note that $q$ is the law after at least one transition, and every one-step transition is dominated by $2p_i$; hence $q_i\le2p_i$. Also,
\[
    d\max_iq_i\le d\left(\sum_iq_i^2\right)^{1/2}=O(t_1^{-2}).
\]
The term $2t_1p_i$ is controlled by $\max_i p_i\le\sqrt{t_1/t_3}$ under the same condition. Hence, if $t_3>(A t_1)^{13}/4$ and $A=A(\delta,\Delta)$ is sufficiently large, then
\[
    \max_i\{2t_1p_i+d(q_i+\varepsilon)\}
    +4t_1^{13/4}t_3^{-1/4}
    \le \frac{1}{143},
\]
after also increasing $A$ to cover finitely many small $N$. Therefore
\begin{equation}\label{eq:recent-repeat-bound}
    \bP(R-R_0\le d)\le \frac{6}{13}.
\end{equation}
On the same event range, \eqref{eq:ER-upper-half} gives
\[
    \frac{\bE (R)}{(A t_1)^{13}}
    \le
    \frac{4t_1t_2}{(A t_1)^{13}}
    \le d
\]
after increasing $A$ if necessary. Thus
\begin{equation}\label{eq:half-column-gap}
    \bP\left(R-R_0<\frac{\bE (R)}{(A t_1)^{13}}\right)\le\frac{6}{13}
\end{equation}
when $t_3>(At_1)^{13}/4$. If $t_3\le(At_1)^{13}/4$, then
\eqref{eq:ER-upper-half} gives $\bE (R)/(At_1)^{13}\le1$, while $R-R_0\ge1$ almost surely, so \eqref{eq:half-column-gap} is trivial.

We now transfer the same bound to the row process without adding any assumption. Let $\overline M=(1-M)^\top$. The column process of $\overline M$ is the row process $(b_k)$ of $M$, shifted by one half-step: from a row of $M$ it first chooses a zero in that row and then a one in the resulting column. If $M$ satisfies $\rmax\le\delta n$, then every column of $\overline M$ has sum at least $(1-\delta)n$, so $\overline M$ satisfies the dense-column alternative with parameter $1-\delta>1/2$. If $M$ satisfies $\cmin\ge\Delta m$, then every row of $\overline M$ has sum at most $(1-\Delta)m$, so $\overline M$ satisfies the sparse-row alternative with parameter $1-\Delta<1/2$. In both cases the transformed chain has contraction parameter at least the $\eta$ used above, so the same block length $t_1$ is valid. Applying \eqref{eq:half-column-gap} to $\overline M$ gives
\begin{equation}\label{eq:half-row-gap}
    \bP\left(R'-R'_0<\frac{\bE (R')}{(A t_1)^{13}}\right)\le\frac{6}{13},
\end{equation}
with $A$ still depending only on $\delta,\Delta$.

Finally, $C\asymp\min(R,R')$ and
\[
    L\ge 2\min(R-R_0,R'-R'_0).
\]
Since $\bE(\min(R,R'))\le \bE(R)$ and $\bE(\min(R,R'))\le \bE (R')$, \eqref{eq:half-column-gap} and \eqref{eq:half-row-gap} imply
\[
\begin{aligned}
    \bE (L)
    &\ge
    \frac{2\,\bE\min(R,R')}{(A t_1)^{13}}
    \bP\left(
    \min(R-R_0,R'-R'_0)
    \ge \frac{\bE\min(R,R')}{(A t_1)^{13}}
    \right) \\
    &\ge
    \frac{2}{13(A t_1)^{13}}\bE(\min(R,R')).
\end{aligned}
\]
Because $t_1=O_{\delta,\Delta}(\log N)$ and $C\asymp\min(R,R')$,
\[
    \frac{\bE (C)}{\bE (L)}
    =O_{\delta,\Delta}(\log^{13}N).
\]
This proves the theorem.
\end{proof}
	
\subsection{Proof of Theorem \ref{thm:perm-mixing}}\label{app:proof-perm-mixing}
We identify $\Sigma(\mathbf{1}_n,\mathbf{1}_n)$ with the symmetric group $S_n$ by writing the unique one in row $i$ at column $\sigma(i)$. Under this identification, one step of the \snake algorithm is right multiplication by a random cycle with the following law. Let $X_1$ be uniform on $[n]:=\{1,\ldots,n\}$. Given $X_s$, sample $X_{s+1}$ uniformly from $[n]\setminus\{X_s\}$ until the first time $T$ at which $X_T\in\{X_1,\ldots,X_{T-2}\}$. If $X_T=X_J$, then the generated cycle is
\[
    (X_J,X_{J+1},\ldots,X_{T-1}).
\]
Let \(\Lambda:=T-J\) denote the length of the generated permutation cycle. The corresponding
\snake transition flips $2\Lambda$ matrix entries.

The distribution is invariant under relabeling of $[n]$. Hence, conditional on the cycle length being $\ell$, the generated cycle is uniform over the conjugacy class of all $\ell$-cycles. We write
\[
    P_n=\sum_{\ell=2}^n p_\ell\,\mu_\ell,
\]
where $p_\ell$ is the probability that the generated cycle has length $\ell$, and $\mu_\ell$ is the uniform law on $\ell$-cycles.

\begin{lemma}[Birthday scale of the snake cycle]\label{lem:perm-birthday}
Let $T$ be the number of raw labels sampled in one step, and let $\Lambda$ be the length of the generated cycle. For $2\le k\le n$,
\[
    \bP(T\ge k+1)=\frac{(n-1)(n-2)\cdots(n-k+1)}{(n-1)^{k-1}}.
\]
Moreover,
\[
    \bE (T)=\sqrt{\frac{\pi n}{2}}+O(1),\qquad
    \bE (\Lambda)=\sqrt{\frac{\pi n}{8}}+O(1),
\]
and there are constants $0<a_1<a_2<\infty$ and $p_0>0$, independent of $n$, such that
\[
    \bP(a_1\sqrt n\le \Lambda\le a_2\sqrt n)\ge p_0.
\]
\end{lemma}

\begin{proof}
The event $\{T\ge k+1\}$ is the event that $X_1,\ldots,X_k$ are all distinct. The first label has $n$ choices, while every later label has $n-1$ possible choices because immediate repetition is forbidden. Thus
\[
    \bP(T\ge k+1)
    =\frac{n(n-1)\cdots(n-k+1)}{n(n-1)^{k-1}}
    =\frac{(n-1)(n-2)\cdots(n-k+1)}{(n-1)^{k-1}}.
\]
Therefore
\[
    \bP(T\ge \lfloor c\sqrt n\rfloor+1)\to \exp(-c^2/2)
\]
for each fixed $c>0$, by the standard birthday-problem estimate. Summing the tail probabilities, or equivalently using the Ramanujan $Q$-function asymptotic in \citep{Flajolet2}, gives $\bE (T)=\sqrt{\pi n/2}+O(1)$.

Conditional on $T=k+1$, the repeated label is uniform over the first $k-1$ labels that are not the immediately previous one, and the resulting cycle length is uniform on $\{2,\ldots,k\}$. Hence $\bE (\Lambda)=\bE (T)/2+O(1)$, which gives the displayed expectation. The same birthday estimate gives fixed constants $0<b_1<b_2$ such that $\bP(b_1\sqrt n\le T\le b_2\sqrt n)\ge 3/4$ for all large $n$. Conditional on $T\ge b_1\sqrt n$, the chance that $\Lambda\ge b_1\sqrt n/3$ is bounded below by a positive absolute constant. Taking $a_1=b_1/3$ and $a_2=b_2$ proves the last claim, after adjusting constants for small $n$ if needed.
\end{proof}

\begin{lemma}[Lower bound]\label{lem:perm-lower}
There is a constant $c>0$ such that, for all $t\le c\sqrt n\log n$,
\[
    \left\|P_n^t-U(S_n)\right\|_{\TV}\ge \frac14+o(1).
\]
\end{lemma}

\begin{proof}
It is enough to consider the walk started from the identity. Let $N_t$ be the total number of raw labels sampled in the first $t$ steps, so $N_t=T_1+\cdots+T_t$ with $T_i$ distributed as in Lemma \ref{lem:perm-birthday}. If $t\le c\sqrt n\log n$, then
\[
    \bE (N_t)\le c\sqrt n\log n\left(\sqrt{\frac{\pi n}{2}}+O(1)\right).
\]
Choose $c>0$ so small that $\bP(N_t>n\log n/2)\le 0.01+o(1)$ by Markov's inequality.

Let $\tau$ be the first raw sampling time by which all labels in $[n]$ have appeared. The usual coupon-collector proof, unchanged by the rule forbidding immediate repetition, gives $\tau/(n\log n)\to 1$ in probability: after $i-1$ distinct labels have appeared, the next raw label has probability between $(n-i+1)/n$ and $(n-i+1)/(n-1)$ of being new. Hence $\bP(\tau\le n\log n/2)=o(1)$.

On the event $\{N_t\le n\log n/2<\tau\}$, some label has not appeared in any generated cycle, so that label is still a fixed point of the current permutation. Let $A$ be the event that a permutation has at least one fixed point. We have
\[
    \bP(\sigma_t\in A)\ge 0.99+o(1).
\]
Under the uniform distribution on $S_n$, the derangement formula gives $U(S_n)(A)=1-1/e+o(1)$. Therefore
\[
    \left\|P_n^t-U(S_n)\right\|_{\TV}
    \ge \bP(\sigma_t\in A)-U(S_n)(A)
    \ge \frac1e-0.01+o(1)>\frac14+o(1),
\]
which proves the lower bound.
\end{proof}

The upper bound uses standard Fourier analysis on $S_n$. We recall only the facts needed here. Irreducible representations of $S_n$ are indexed by partitions $\lambda\vdash n$. Let $d_\lambda$ be the dimension, and let
\[
    r_\ell(\lambda)=\frac{\chi^\lambda(C_\ell)}{d_\lambda}
\]
be the character ratio at the conjugacy class $C_\ell$ of $\ell$-cycles. If $Q$ is a probability measure on $S_n$ that is constant on conjugacy classes, then the Fourier transform of $Q$ in representation $\lambda$ is scalar:
\[
    \rho_\lambda(Q)=\left(\sum_\ell q_\ell r_\ell(\lambda)\right)I_{d_\lambda},
\]
where $q_\ell$ is the total $Q$-mass on $\ell$-cycles. The Plancherel bound used by \citet{diaconis1981generating} gives
\begin{equation}\label{eq:perm-plancherel}
    4\|Q^s-U(S_n)\|_{\TV}^2
    \le \sum_{\lambda\ne \mathrm{triv}} d_\lambda^2
    \left|\sum_\ell q_\ell r_\ell(\lambda)\right|^{2s}.
\end{equation}
The trivial representation is omitted. The sign representation has $r_\ell(\mathrm{sgn})=(-1)^{\ell-1}$.

\begin{lemma}[Parity smoothing]\label{lem:perm-sign}
Let $I_n=\{\ell:\lceil a_1\sqrt n\rceil\le \ell\le \lfloor a_2\sqrt n\rfloor\}$, where $a_1,a_2$ are fixed positive constants. If
\[
    p_I:=\sum_{\ell\in I_n}p_\ell
\]
is bounded below by a positive constant, then
\[
    \left|\sum_{\ell\in I_n}(-1)^{\ell-1}\frac{p_\ell}{p_I}\right|=O(n^{-1/2}).
\]
\end{lemma}

\begin{proof}
For $\ell\ge 2$, conditioning on the raw stopping time gives
\[
    p_\ell-p_{\ell+1}=\frac{\bP(T=\ell+1)}{\ell-1}.
\]
Indeed, for $T\ge \ell+2$, the conditional probabilities of producing an $\ell$-cycle and an $(\ell+1)$-cycle are equal; the only unmatched case is $T=\ell+1$. Also $p_\ell\le 1/(\ell-1)=O(n^{-1/2})$ uniformly over $\ell\in I_n$, and
\[
    \sum_{\ell\in I_n}|p_\ell-p_{\ell+1}|
    \le \sum_{\ell\in I_n}\frac{\bP(T=\ell+1)}{\ell-1}
    =O(n^{-1/2}).
\]
Pairing consecutive terms in the alternating sum over $I_n$ leaves only endpoint terms and the above total variation of the sequence $(p_\ell)$ over $I_n$. Hence $\left|\sum_{\ell\in I_n}(-1)^{\ell-1}p_\ell\right|=O(n^{-1/2})$. Dividing by $p_I$ proves the lemma.
\end{proof}

We also use two estimates of \citet{hough2016random}, whose second estimate relies on the dimension bounds of \citet{larsen2008characters}. There are constants $\delta>0,c_1>0,c_2>0$ such that, uniformly in $n$, every non-trivial, non-sign irreducible representation $\lambda$, and every $2\le \ell\le \delta n$,
\begin{equation}\label{eq:perm-hough-ratio}
    |r_\ell(\lambda)|^{(n/\ell)(\log n+c_1)}\le d_\lambda^{-1},
\end{equation}
and
\begin{equation}\label{eq:perm-hough-dim}
    \sum_{\lambda\vdash n}d_\lambda^{-c_2/\log n}=O(1).
\end{equation}
Apart from the trivial and sign representations, every irreducible representation of $S_n$ has dimension at least $n-1$.

\begin{lemma}[Upper bound for good cycle lengths]\label{lem:perm-good-upper}
Let $P_{n,I}$ be the law $P_n$ conditioned on $\Lambda\in I_n$, with $I_n$ chosen as in Lemma \ref{lem:perm-birthday}; in particular $p_I\ge p_0>0$. For every fixed $C>0$, there is a constant $A_C<\infty$ such that
\[
    \|P_{n,I}^s-U(S_n)\|_{\TV}=O(e^{-C/2})
\]
whenever $s\ge A_C\sqrt n\log n$.
\end{lemma}

\begin{proof}
For $n$ large, every $\ell\in I_n$ is at most $\delta n$. By \eqref{eq:perm-hough-ratio}, for every non-trivial, non-sign $\lambda$,
\[
    \left|\sum_{\ell\in I_n}\frac{p_\ell}{p_I}r_\ell(\lambda)\right|
    \le d_\lambda^{-a_1/(\sqrt n(\log n+c_1))}.
\]
Take
\[
    s_0=\frac{\sqrt n}{a_1}(\log n+c_1)
    \left(1+\frac{c_2+C}{2\log n}\right).
\]
Then, for all $s\ge s_0$,
\begin{align*}
    \sum_{\lambda\notin\{\mathrm{triv},\mathrm{sgn}\}}
    d_\lambda^2
    \left|\sum_{\ell\in I_n}\frac{p_\ell}{p_I}r_\ell(\lambda)\right|^{2s}
    &\le \sum_{\lambda\notin\{\mathrm{triv},\mathrm{sgn}\}}
    d_\lambda^{-(c_2+C)/\log n} \\
    &\le (n-1)^{-C/\log n}
    \sum_{\lambda\vdash n}d_\lambda^{-c_2/\log n}
    =O(e^{-C}).
\end{align*}
By Lemma \ref{lem:perm-sign}, the sign-representation contribution in \eqref{eq:perm-plancherel} is at most $O(n^{-1/2})^{2s}=o(e^{-C})$. Applying \eqref{eq:perm-plancherel} proves the result.
\end{proof}

\begin{lemma}[Upper bound for the original walk]\label{lem:perm-upper}
For every fixed $C>0$, there is a constant $B_C<\infty$ such that
\[
    \|P_n^t-U(S_n)\|_{\TV}=O(e^{-C/2})
\]
whenever $t\ge B_C\sqrt n\log n$.
\end{lemma}

\begin{proof}
Write
\[
    P_n=p_I P_{n,I}+(1-p_I)P_{n,I^c},
\]
where $p_I\ge p_0>0$. In $t$ steps, let $G_t$ be the number of times the first component $P_{n,I}$ is selected. Then $G_t\sim\Binom(t,p_I)$. Fix a word of length $t$ containing $s$ copies of $P_{n,I}$. Since both $P_{n,I}$ and $P_{n,I^c}$ are class functions, all Fourier transforms are scalar, and the same Plancherel bound as above gives a TV bound no larger than the corresponding bound for $P_{n,I}^s$; the factors from $P_{n,I^c}$ have absolute value at most one. Hence, by convexity of TV distance,
\[
    \|P_n^t-U(S_n)\|_{\TV}
    \le \bP(G_t<s_0)+O(e^{-C/2}),
\]
where $s_0$ is the threshold in the proof of Lemma \ref{lem:perm-good-upper}. If $t\ge 2s_0/p_0$, then $\bE (G_t)\ge 2s_0$, and Chernoff's bound gives $\bP(G_t<s_0)\le \exp(-s_0/4)=o(e^{-C/2})$ for fixed $C$. This proves the upper bound.
\end{proof}

\begin{proof}[Proof of Theorem \ref{thm:perm-mixing}]
Lemma \ref{lem:perm-lower} gives $t_{\mathrm{mix}}^{\mathrm{perm}}(1/4)\ge c\sqrt n\log n$ for a positive constant $c$. Lemma \ref{lem:perm-upper}, with $C$ chosen sufficiently large, gives $t_{\mathrm{mix}}^{\mathrm{perm}}(1/4)\le C'\sqrt n\log n$ for a finite constant $C'$. Therefore $t_{\mathrm{mix}}^{\mathrm{perm}}(1/4)=\Theta(\sqrt n\log n)$.

The extension from the threshold $1/4$ to any fixed $0<\epsilon<1/2$ follows from the standard equivalence of fixed total-variation thresholds for finite Markov chains; see, for example, the submultiplicativity inequalities in \cite[Chapter~4]{LP}. This completes the proof.
\end{proof}

\subsection{Proof of Theorem \ref{thm: rapid mixing}}\label{app:proof-rapid-mixing}
\begin{proof}
If $|\Sigma(\row,\col)|=1$, the result is immediate under the singleton spectral-gap convention. Hence assume $|\Sigma(\row,\col)|\ge2$.
Let $P_{\swap}$ and $P_\snake$ be the raw kernels, let $P_{\swap,L}$ and $P_{\snake,L}$ be their lazy versions, put
\[
\mathcal M=\binom{m}{2}\binom{n}{2},
\qquad
\eta_{m,n}=\frac{1}{mn\max\{m,n\}^3}.
\]
If $A$ and $B$ differ by one checkerboard swap, then the \snake chain can move from $A$ to $B$ by selecting one corner of the checkerboard as the starting entry and then selecting the remaining three corners in their alternating order. The starting entry has probability $1/(mn)$, and each of the three conditional choices has probability at least $1/\max\{m,n\}$. Hence
\[
P_{\snake,L}(A,B)\ge \frac{\eta_{m,n}}{2}.
\]
On the other hand, the lazy swap chain chooses this particular row pair and column pair and takes its non-lazy half-step, so
\[
P_{\swap,L}(A,B)=\frac{1}{2\mathcal M}.
\]
It follows that, for every off-diagonal pair,
\[
P_{\snake,L}(A,B)\ge
\eta_{m,n}\mathcal M\,P_{\swap,L}(A,B);
\]
when $A$ and $B$ are not swap-adjacent, the right-hand side is zero.

Both lazy chains are reversible with respect to the uniform distribution on $\Sigma(\row,\col)$. Hence their Dirichlet forms satisfy
\[
\mathcal E_{\snake,L}(h,h)
=\frac{1}{2}\sum_{A,B}U(A)P_{\snake,L}(A,B)(h(A)-h(B))^2
\ge \eta_{m,n}\mathcal M\,\mathcal E_{\swap,L}(h,h)
\]
for every real-valued function $h$ on $\Sigma(\row,\col)$. It follows from the variational formula for the spectral gap that
\[
\gap(P_{\snake,L})\ge
\eta_{m,n}\mathcal M\,\gap(P_{\swap,L}).
\]
By the universal swap-chain estimate of \citet{fu2026spectral},
\[
    \gap(P_{\swap,L})\ge \mathcal M^{-1}.
\]
Therefore $\gap(P_{\snake,L})\ge\eta_{m,n}$, proving the claimed spectral-gap bound. Finally, the standard reversible-chain inequality
\[
t_{\mix}^{\snake}(\epsilon)
\leq \gap(P_{\snake,L})^{-1}
\left(\log \frac{1}{\pi_{\min}}+\log\epsilon^{-1}\right)
\]
and the uniform stationary probabilities
$\pi_{\min}=|\Sigma(\row,\col)|^{-1}$ give the first mixing-time bound. The second follows from
$|\Sigma(\row,\col)|\leq 2^{mn}$.
\end{proof}

	\newpage
\setstretch{1.24}

\end{document}